\documentclass[11pt, letter]{article}
\usepackage[margin=0.8in]{geometry}

\usepackage{latexsym, amssymb, amsfonts, amsmath, amsthm, graphics, graphicx, bbm}
\usepackage{subfig}
\usepackage{epstopdf}
\usepackage[colorlinks,linkcolor=blue,menucolor=red,backref=false,bookmarks=true]{hyperref}

\usepackage{palatino}
\usepackage{multicol}
\usepackage{tikz}
\usetikzlibrary{shapes,arrows,positioning}
\usepackage[T1]{fontenc}

\newtheorem{theorem}{Theorem}

\newtheorem{definition}[theorem]{Definition}
\newtheorem{example}[theorem]{Example}

\newtheorem{lemma}[theorem]{Lemma}

\newtheorem{proposition}[theorem]{Proposition}
\newtheorem{remark}[theorem]{Remark}

\newcommand{\field}[1]{\mathbb{#1}}
\newcommand{\R}{\field{R}}

\newcommand{\E}{\field{E}}

\newcommand{\1}{\field{1}}
\newcommand{\p}{\field{P}}

\newcommand{\mat}[2][rrrrrrrrrrrrrrrrrrrrrrrrrrrrrrrrrrrrrrrrrrrrrrrrrrr]{\left[ \begin{array}{#1} #2 \\ \end{array}\right]}
\newcommand{\met}[2][ccccccccccccccccccccccccccccccccc]{\left[ \begin{array}{#1} #2 \\ \end{array}\right]}

\newcommand{\ms}{\bar \sigma}
\newcommand{\eps}{\epsilon}

\DeclareMathOperator{\sign}{sign}

\numberwithin{equation}{section}

\title{Averaging + Learning Models and Their Asymptotics}
\author{
    Ionel Popescu \thanks{I.P. was partially supported by UEFISCDI PN-III-P4-ID-PCE-2016-0372}\\
  University of Bucharest \\ Institute of Mathematics "Simion Stoilow" of the Romanian Academy\\
  \url{ionel.popescu@fmi.unibuc.ro} \normalsize
\\ 
  Tushar Vaidya \\
  SUTD \\
  \url{tushar.vaidya@ntu.edu.sg} \normalsize
}

\begin{document}
\maketitle

\begin{abstract}
We develop original models to study interacting agents in financial markets and in social networks. Within these models randomness is vital as a form of shock or news that decays with time. Agents learn from their observations and learning ability to interpret news or private information in time-varying networks. Under general assumptions on the noise, a limit theorem is developed for the generalized Averaging framework for certain type of conditions governing the learning.  In this context, the agents' beliefs (properly scaled) converge in distribution that is not necessarily normal.  Fresh insights are gained not only from proposing a new setting for social learning models but also from using different techniques to study discrete-time random linear dynamical systems.
\end{abstract}

\tableofcontents
\section{Introduction}
How do agents reach consensus on prices in realistic noisy environments? This is the central theme of this paper. Traders interact with one another and \emph{learn} from their own environment. While learning and interaction can be modelled satisfactorily, the introduction of noise complicates the choice of a framework and model. Furthermore, the choice of discrete or continuous time expands this issue. Our aim is to propose new models of interaction and learning where noise is a significant feature in discrete time. These new models of learning and interaction entail agents who observe actions of other traders and update their own beliefs. Repeated interaction can in certain cases lead to consensus on a particular value of a tradeable commodity. Interaction models should take into account the environment of trading. When agents interact in financial markets, they are visually, algorithmically observing each other other's quotes.

Prices are updated fast and in a simultaneous manner but the environment is always noisy. Averaging models offer a cogent and natural way to analyse interaction when agents learn and observe each others' past actions through an online platform. Visual screens provide a snapshot of where current prices are.  For such models there is a rich interplay between probability, dynamical systems and game theoretic ideas \cite{mossel2017opinion}.  
In financial markets, trading is never sequential (one player at a time) nor is information perfectly perceived or received by agents. Transactions occurs at breakneck speed \cite{buchanan2015physics}. Agents (algorithms) move simultaneously: cancellations are the norm in today's fast markets. In practical terms, sequential learning models don't seem appropriate \cite{vaidya2020learning, SequentialNaiveLearningArieli}. Interaction is important in the emergence of consensus. Choices by agents from the previous round of play are available to all agents in the current round of play. The question is then what sort of averaging or heuristic process is ideal?

DeGroot type of learning models convey an essential and robust idea that is taking a firmer foothold in theory \cite{amir2021robust,SequentialNaiveLearningArieli,molavi2018theory}. They offer a functional form of updating. Myopic updating occurs in each round. Something akin to persuasion bias could explain our basic model \cite{demarzo2003persuasion}. As in an echo chamber, agents in our setup have fixed weights but update their responses until consensus is reached. One could think of it as a behavioural heuristic and why repeated averaging is effective. Alternatively, with the right cost function representing the distance of an agent's opinion against other opinions, the best response is repeated averaging. Recently there have been some experimental papers on evidence of DeGroot updating \cite{chandrasekhar2020testing, becker2017network}. Repeated averaging models are our base precisely because they capture the nature of interaction and learning in financial markets so succinctly. On top of the base models we develop more sophisticated extensions, relaxing the fixed nature of the weights and learning matrices.

Our discussion so far is on trading but any network where the players have access to some sort of learning feedback is suitable.  A game theoretic framework where every player takes into account other players' payoff is unrealistic and points to serious difficulties on how to even represent utilities \cite{kirman2002reflections}; these are economic arguments that are better addressed by in depth philosophical interludes.  Moreover, traders rarely have access to private information on how previous decisions led to a certain payoff for their opponent at least not in a high frequency sense. If a trading firm is a publicly listed company, then one can infer its trading losses or gains from public records.  Nevertheless specific profit and loss accounts of trading individual stocks is a private matter. Firms never break down their income statements down to specific asset classes or instruments. Results are amalgamated and reported quarterly: not per hour, minute or second. Dynamical learning is an active area of research in computer science as well. Articles \cite{papadimitriou2018game,piliourasaamas, mai2017opinion} propose and analyze the dynamics rather than just the concept of Nash equilibrium. Similar questions and issues to this paper were raised in \cite{kirman2002reflections} at an informal level. 


\subsection{Basic models - Averaging + Learning}	
There are many models of noisy learning \cite{acemoglu2011opinion, blondel2005convergence, lebaron2001builder, masuda2017random, mueller2013general, Sobel2000EconomistsMO,  sobkowicz2020whither}. Depending on the question, different paradigms have been put forth. Our objective is learning and so we aim to use aspects of both social learning theory and dynamical systems. Difficulties in Bayesian environments mean the DeGroot model has become a workhorse for social learning \cite{banerjee2019naive, Sobel2000EconomistsMO}. It offers a way forward for tractable models that can relax simple assumptions. Research using this framework is still active. In our setting, a group of traders  observe quotes of others and incorporate an average of the previous quotes. The stark departure from standard DeGroot learning comes from the fact that not only are the agents averaging but they are getting imperfect feedback from an external/internal source on the true consensus value. To our knowledge, the setting of these types of consensus models to trading is new. We use the framework of 
\cite{vaidya2018learning, vaidya2020learning} as the base case for our models,  generalizing these models in the form

\begin{equation}\label{i:0}
x_t^i=\underbrace{f_{i,t}(X_{t-1})}_{\color{red} averaging} +\underbrace{g_{i,t}(x_{t-1}^i)}_{\color{blue} learning},
\end{equation}	
where $X_t$ is the vector of previous beliefs of $n$ agents at time $t$ and $f_{i,t},g_{i,t}$ determine the dynamic for the $i$th agent. We treat the case where $f$ is linear and $g$ could be allowed to be non-linear.  To keep things simple  we will discuss in the introduction the full linear case given by
\begin{equation}\label{eqn:dt_opinion_dyn}
	x^i_{t+1}=\sum_{j=1}^n (a_t)_{ij}x^j_{t}+\epsilon_{i,t}(\sigma-x^i_{t}), 
\end{equation}
which in the matrix form reads
\begin{equation}
X_{t+1}=A_tX_t+ \mathcal{E}_t(\sigma \mathbbm{1}_n - X_t),
\end{equation}
where $X_t = (x^1_t,...,x^n_t)^T$ is the opinion of each agent in discrete-time  $t$, and $\mathcal{E} = \textrm{diag}(\epsilon_1,...,\epsilon_n)$ is the learning rate of each agent when they are provided with a known feedback $\sigma$. The stochastic matrix $A_t$ encapsulates the weights agents put on each other at time $t$ and is partially similar to the DeGroot model.

Agents' aptitude to determine the quality of feedback is their ability $\epsilon_i$. A concrete example is in the options market. A trader observes quotes for volatility of an option in the market. At the same time, the option (with a particular strike and maturity) will have a theoretical known value of $\sigma$. The trader's own ability to extract this theoretical model is conveyed by his or her institution's quantitative model. Options, like other financial assets, have a replicable, theoretical value implied by a model. This theoretical value means that the feedback is interpreted as skill by the trader's institution. 

We should point out from the beginning that this model is \emph{not} intended to be a DeGroot model.  Our focus here is to study more conditions under which consensus arises.  Thus, in this first round, the value of $\sigma$ is known to all players and the goal is to give conditions under which the interaction leads to convergence to the real value of $\sigma$.  A \emph{key} difference between our model and existing literature is that recursions of the form $X_t=A_t X_{t-1}+ B_t$ in our setup have time dependent matrices. For systems with independent and identically distributed $A_t,\, B_t$ there has been prior work \cite{de2004multidimensional, diaconis1999iterated, buraczewski2010convergence, buraczewski2016stochastic, kesten1973random, parasnis2022random}.  In traditional game-theoretic models the focus is on equilibrium \cite{fudenberg1998theory}. The focus here is how agents \textbf{behave} in the long run, either they converge to a point or (properly scaled) to a distribution.

\subsection{General Noisy Models} The base model is deterministic.  The next natural and more realistic extension is to assume that the value of $\sigma$ is noisy.  In this context, the model has a more dynamical interpretation, the agents do not observe the true value of $\sigma$, however our main focus is to study conditions under which the agents converge in some form. 

We study this under various forms.  The first type is the one in which the noise decreases with time to $0$.  The other case is where the noise is persistent.  We give some general conditions under which the convergence takes place. Under these general conditions on the interaction matrix and the learning rates, if the noise is decaying to $0$ in some mathematical form, then the agents converge to $\sigma$.  If the noise is permanent though, then, under some supplementary conditions, the convergence is in a distributional sense. 

It is worth to mention, how we quantify this. The key point is to define
\[
\rho_t=\max_{i=1,\cdots,n} (|(a_t)_{ii}-\eps_{i,t}|+1-(a_t)_{ii}).
\]
Notice that this is easy to compute and depends only on the diagonal part of the  matrix $A_t$ and the learning rates matrix $\mathcal{E}_t$.  Our condition on convergence is driven by the following ensemble of these individual quantities as 
\[
 \sup_{t\ge1}\{\rho_t+\rho_t\rho_{t-1}+\rho_t\rho_{t-1}\rho_{t-2}+\dots+\rho_{t}\rho_{t-1}\dots\rho_1\}<\infty.
\]
Notice that in the pure DeGroot model, $\rho_t=1$ for all $t$, the above condition is not satisfied which shows that the learning part plays an important role in the convergence.  

Under another condition on the fluctuation of the matrices $A_t$ and $\mathcal{E}_t$, we show that if the noise in the feedback is iid, then the distribution of $X_t$ converges.  However, the limiting distribution is not easy to describe.  We point out that in the case the learning part $g$ in the model \eqref{i:0} is allowed to be non-linear, we can give a version of the quantity $\rho_t$ under which most of the results can be extended.

\subsection{Limiting distributions} 
We prove that when the noise is iid, under some conditions, we have convergence in distribution of the agents' opinions.  However, the description of the limiting distribution is not explicit.  To understand a little more the limiting distribution, we take the linear case where the matrices $A$ and $\mathcal{E}$ are time independent.  Here we can prove that after rescaling the vector $X_t$, we can compute the asymptotic behavior.  This asymptotic behavior is fully controlled by the Jordan decomposition of the matrix $A-\mathcal{E}$.  

One of the possible behavior of the vector $X_t$ in the limit is governed by the normal distribution, very similar to the case of the Central Limit Theorem.  Although, this is one of the several possibilities; the one in which some of the eigenvalues of $A-\mathcal{E}$ are on the unit circle and all the rest are inside the unit disk.  In this case, the behavior is as one would expect, the noise appears in the limit through its covariance matrix and not its distribution.

What is really interesting to point out is that in addition to this behavior there are other behaviors.  Indeed, if all the eigenvalues of $A-\epsilon$ are in the unit disk, then the vector $X_t$ converges without any additional scaling.  However, the limit depends on the whole distribution of the noise.  

In the remaining case, if the matrix $A-\mathcal{E}$ has eigenvalues of absolute value strictly bigger than 1, then $X_t$ grows exponentially and in addition could also oscillate.  Overall, the whole distribution of the noise is involved in the limiting behavior of $X_t$.  

For the general limiting behavior of $X_t$, we can observe in \autoref{s:cltsim} and \autoref{f:alpha1andbigCLT1} distributions that could be fractal like depending on the noise. Important though is to mention that these limiting distributions could be used from the statistical standpoint to give estimates on the true parameter $\sigma$. 

It is also worth emphasizing that the limit distribution of the scaled dynamic is extremely sensitive to the structure of the Jordan decomposition of the matrix $A-\mathcal{E}$ and in particular the spectrum of the matrix $A-\mathcal{E}$.  Therefore, slight perturbation of the matrix could lead to radically different behavior of $X_t$ particularly in the case of eigenvalues of $A-\mathcal{E}$ in the proximity of the unit circle.


\subsection{Noisy Feedback and Related Literature}
The feedback can best explain the situation where a similar instrument is traded on another exchange or there is a common source of market chatter. 
Moreover, such chatter is commonly provided through voice box brokers 
or over-the-counter markets. We assume all agents have access to this 
feedback or chatter.  Take for example trading in Foreign Exchange (FX). Trading in USD/Euro can occur in both London and New York, when both markets are open. During this period, traders from both markets can see what each group is quoting, and during the crossover period before the New York traders become dominant, quotes from both centres will be tightly linked. Their ability via $\varepsilon$ to capture feedback on this $\bar{\sigma}$ through either external markets or internal models is represented by $\epsilon_i(\sigma-x^i_{t})$. In reality, for G7 currencies, most of the spot trading will be done by algorithms tracking each other's movements. At the same time these algorithms will have internal models guiding them to what the true theoretical value of the asset is, ensuring the price quoted is arbitrage free. Thus, feedback or learning is aptly represented and illustrated as a separate process to visible, external averaging. Another example would the S\&P500 European ETF (SPY) options, which are not cash settled as SPX options but stock settled.  
Quotes for the SPY options will also be linked with the SPX options. There will always be a model or arbitrage-free implied theoretical value. Equities are traded across multiple exchanges and platforms like currencies. Sometimes trades occur off-exchange and get reported at the end of the day through the Exchange's clearing system. This points to the fact that models with Averaging + Learning aspects are well suited for portraying trading dynamics as social phenomenon for a different types of financial assets. How agents interpret information or market chatter is their unique learning ability. This ability is impacted by noise. Learning is not \emph{perfect}. No additional assumption is required apart from the fact that noise is either persistent or transient. In either case, the system settles down, converging to consensus or to a distribution. Noise can be either common to all players or distinct. 


\subsubsection*{Related Literature}
Previous literature has not addressed this noisy feedback in a sustained, systematic effort to investigate averaging models, partly because of the difficulties of discrete-time dynamics and asymptotic analysis. There have been some contributions focusing on empirical results \cite{kozma2008consensus, monin2021information}. While noise was added to a pure DeGroot model in \cite{stern2021impact}, a Central Limit Theorem was not proven, rather noise added was assumed to be Gaussian; though this work did introduce \textit{global} and \textit{local} uniqueness noise models in its framework through simulations. Noise is also added in bounded confidence models \cite{su2020noise} such as Hegselmann-Krause (HK) models, which are quite distinct from DeGroot dynamics \cite{hegselmann2002opinion, su2017noise, wang2017noisy}. Interestingly, but complementary to our work is that of $\frac{1}{m}$-DeGroot dynamics \cite{amir2021robust}. A variant of DeGroot dynamics that is robust to stubborn agents and mis-specification is introduced. However, different from our framework, noise is present only at the beginning of the dynamics, after which the focus is on the interaction protocol with a large number of agents. This type of $\frac{1}{m}$-DeGroot dynamics approximates the standard DeGroot dynamics to the nearest fraction of the average of the players. Agents are restricted to choosing an opinion on the $\frac{1}{m}$ grid. It is a coarsening of the DeGroot dynamics. This model is in response to the vulnerability of stubborn agents leading consensus astray \cite{acemouglu2013opinion, golub2010naive, ghaderi2014opinion}. In our framework, certainly agents can be stubborn and not interact. Thus the dynamics in \ref{eqn:dt_opinion_dyn} can feature a row with $0$'s and one $1$ for the agent not interacting or sticking to her previous opinion. Still, in our models the learning aspect through $\mathcal{\varepsilon}$ ensures every agent reaches asymptotic consensus. Unlike some Bayesian models, noise in our framework can be common or distinct to each agent \cite{mossel2014asymptotic, hkazla2019reasoning}. Moreover, unlike Kinetic models of opinion exchange \cite{during2009boltzmann, during2015opinion, toscani2006kinetic}, our models are in discrete-time so continuous-time techniques are not applicable. Time-varying networks have been studied in continuous-time frameworks in different settings \cite{carletti2022theory, ghosh2022synchronized, mcquade2019social}. However, our concept of noise in space and time is quite different from earlier works. 

The DeGroot model arises as a special case in our framework, though our results do not apply to the pure DeGroot model. Noise in some Bayesian, social learning models is viewed as a private signal \cite{banerjee1992simple, banerjee2019naive, hkazla2019reasoning}.  However, the emphasis we put is on the probabilistic notions of consensus and realistic computations by agents as $t \to \infty$. This holds regardless of the number of agents.  Our setting is general enough to incorporate a range of models in the feedback. For time varying models with just $A_t$, without noise, see \cite{FB-LNS,lu2008synchronization, shang2022resilient, fagnani2018introduction, weber2019deterministic}. The recursions involving time-varying $A_t$ and $\mathcal{E}_t$, mean we are dealing with backward matrix products. However, the conditions we use for convergence are slightly weaker than even weak ergodicity \cite{bremaud2013markov,chatterjee1977towards}.

\subsection{Organization of the paper}
In \autoref{s:notations} we introduce the main definitions and notations.  Then, in  \autoref{s:BaseModel} we discuss the base model, an extension of the result from \cite{vaidya2018learning}. In \autoref{s:LearningwithRandomnoise} we introduce the models with noise. Theorem \ref{thm:blwn} is the main result of this part.   We treat the non-homogeneous and nonrandom model for the matrix $A_t$ and add external noise in $n$-dimensional Euclidean space.   

We provide in \autoref{s:Nonlinearlearning} a non-linear extension of the linear noisy model. Players still average from their observations of past actions but their own unique learning ability and how they interpret the extra information is a nonlinear function. This type of model fits with the earlier linear models, preserving the averaging nature of interaction. Suitable conditions on the nonlinear function are derived that exhibit consensus. Section \ref{s:CLT} presents the convergence in distribution results in the case of time independent noisy model with iid noise.   For more comments and relevance of the results see \autoref{thm:importance}.


\section{Notation and Assumptions}\label{s:notations}
In all subsequent analysis, $A$ refers to a row-stochastic weights matrix, whose rows sum to one. Depending on the setup, $A$ can be time varying or fixed.

We use the different norms, namely we take for a vector 
$
v=\begin{bmatrix}
	v_1 \\ \vdots \\ v_n
\end{bmatrix}, \, $
\[
|v|_{\infty}=\max_{i=1,\cdots,n}|v_i| \text{ and } |v|_1=\sum_{i=1}^n|v_i| 
\]
and the inner product of two vectors $v,w$ is given by 
\[
\langle v,w\rangle=v'w=\sum_{i=1}^n v_iw_i,  
\]
with the standard use of the transpose for $v'=[v_1,v_2,\dots,v_n]$.   

We have here a duality result computing one norm in terms of the inner product in the form 
\begin{equation}\label{e:na1}
		|v|_\infty=\sup_{|w|_1\le 1} \langle v,w\rangle \text{ and }|v|_1=\sup_{|w|_\infty\le 1}\langle v,w \rangle.  
	\end{equation}

For any $m\times n$ matrix $B$, we denote 
\[
|B|_\infty=\sup_{i=1,\dots, m}\sum_{j=1}^n|b_{ij}| \text{ and } |B|_1=\sup_{j=1,\dots, m}\sum_{i=1}^n|b_{ij}|.  
\]
We then have for any $m\times n$ matrix $B$ and any $n$ dimensional vector $v$
\[
|Bv|_\infty\le |B|_\infty|v|_\infty. 
\]
It is in fact easy to see that 
\[
|B|_\infty=\sup_{|v|_\infty\le 1} |Bv|_\infty.
\]

Finally, for a sequence of $n$-dimensional random variables $(X_t)_{t\ge0}$ 
 and another $n$-dimensional random variable $X$ we use the notation 
\begin{equation}\label{e:convdist}
X_t\Longrightarrow X \text{ if } \,\E[\phi(X_t)]\xrightarrow[t\to\infty]{}\E[X]
\end{equation}
for any continuous and bounded function $\phi:\R^n\to\R$.  


\section{Base Model}\label{s:BaseModel}
We further discuss our base models here to give the reader a perspective of the type of models analyzed. In the base model, we have $n$ agents and a fixed row-stochastic matrix $A$, which is the weights matrix.
Consider the dynamics for updating 

\begin{equation} \label{eq:BM1}
	X_{t+1}=AX_t+ \mathcal{E}(\ms  - X_t).
\end{equation}

We can impose a weaker condition on $\eps_i$ and use $\ms= \sigma 
\mathbbm{1}_n$  for notational convenience when the dimension is clear.  The vector $\mathbbm{1}_n$ denotes the vector with all components equal to $1$. 

\begin{proposition}\label{p:base}
	If $0<\epsilon_i<2 a_{ii}$, then all agents reach the same consensus value, in other words,
	\[\lim_{t \to \infty} X_t = \ms.\]
\end{proposition}

\begin{proof}
Recall the equation \eqref{eq:BM1} of the dynamics 
\[
	X_{t+1}=AX_t+ \mathcal{E}(\ms  - X_t).
\]

Equation \eqref{eq:BM1} can now be rewritten as
\[
X_{t+1}-\ms=(A - \mathcal{E}) (X_t -\ms).
\]
Setting $B=(A - \mathcal{E})$ and $Y_t=X_{t+1}-\ms$, the updating rule simplifies to
\[
(Y_t)_i=\sum_{j=1}^{n}b_{ij} (Y_{t-1})_j,
\] from which we can then obtain
\begin{align*}
	|(Y_t)_i|  &\leq \sum_{j=1}^{n}|b_{ij}| |(Y_{t-1})_j|\\
	&\leq |Y_{t-1}|_\infty \sum_{j=1}^{n}|b_{ij}|.
\end{align*}
Therefore,  $|Y_t|_\infty \leq |Y_{t-1}|_\infty  \max_{\substack{i=1,\cdots,n}}  \sum_{j=1}^{n}|b_{ij}|$.

On the other hand $b_{ij}=a_{ij}$ if $i \neq j$ so that	
\[
\sum_{j=1}^{n}|b_{ij}| =|a_{ii}-\eps_i| +\sum_{j\neq i}^{n}|a_{ij}|  =|a_{ii}-\eps_i| +1 - a_{ii},
\] where we have used the stochasticity of $A$, that is, the sum of the elements of each row is 1.  From this if we check that  $|a_{ii}-\eps_i| +1 - a_{ii} < 1$ which is the same as $|a_{ii}-\eps_i| <a_{ii}$ or equivalently $0<\eps_i<2a_{ii}$, then with 
\[\rho = \max_{\substack{i}} (|a_{ii}-\eps_i| +1 - a_{ii}),
\]
we definitely obtain $0\leq \rho <1$ and  $|Y_t|_\infty \leq \rho \, |Y_{t-1}|_\infty$. This is enough to conclude that 
\[
|Y_t|_\infty \leq \rho^t \, |Y_0|_\infty.
\] From which letting $t \to \infty$ shows that
\[
|Y_t|_\infty\xrightarrow[t\rightarrow\infty]{}0
\] and in particular also proves that  $
Y_t \underset{t \to \infty}{\longrightarrow} 0.
$
\end{proof}

We should point out that in the above result,  the convergence to $\ms$  is exponential and in fact, from the proof, we have that $|X_t-\ms|_\infty\le \rho^t|X_0-\ms|_\infty$.  Thus $\rho$ is a rate of convergence, but it might not be the optimal one.  The true rate of convergence might in fact be much smaller.  This is dictated in principle by the spectral radius which in general is much smaller than the proposed quantity $\rho$ above. This is covered by Gershgorin's theorem.   

For instance, if $A=\mat{1/2 &1/2\\1/2 & 1/2}$, and we take $\epsilon_1=0.01$, while $\epsilon_2=0.99$, then the eigenvalues of $A-\mathcal{E}$ are $\lambda_1=-0.700071, \lambda_2=0.700071$ while $\rho=.99$.  Furthermore, if we take $X_0$ to have equal components equal to $1/2$, then $(X_{100}-\ms)/\lambda_2^{100}=\mat{1/2\\1/3}$ showing that it converges to $0$ much faster.  The result here is a conservative one in the sense that the convergence is still exponential though we do not get the exact rate of convergence.  This analysis works well if the matrix $A$ is time independent, but as soon as we allow $A$ to change with time, the eigenvalue and eigenvector analysis no longer applies.  

For the case of constant matrix, one can have a much better understanding of the convergence rate by simply writing the matrix $A-\mathcal{E}$ in Jordan form as $A-\mathcal{E}=SJS^{-1}$, where $S$ is a matrix of eigenvalues and $D$ is a Jordan block matrix.  From this, one can solve for $X_t=\ms+SJ^t S^{-1}(X_0-\ms)$ and this gives a structure equation for $X_t$ with more details on the behavior of $X_t$ for large $t$.  The decay to $\ms$ is clearly controlled by the eigenvalue with the largest absolute value and its coefficient is given by the corresponding eigenvector.  In the case of eigenvalues with higher multiplicity, we have more contributions but still everything is in terms of the matrices $J$ and $S$.

The argument of Proposition \ref{p:base} allows an extension to the case when the matrices $A_t$ and $\mathcal{E}_t$ depend on t.  The bottom line here is that if we define 
	\[
	\rho_t = \max_{i} (|(a_t)_{ii}-\eps_i(t)| +1 - (a_t)_{ii})
	\] 
 and these quantities satisfy, 
	\begin{equation}\label{eq:BM2}
		\prod_{i=1}^{t} \rho_i \underset{t \to \infty}{\longrightarrow} 0,
	\end{equation}
 then we still have convergence to $\ms$.

For example, this is the case if all $\rho_t$ are bounded by $\rho <1$.  However, condition \eqref{eq:BM2} also allows cases where $\rho_t \underset{t \to \infty}{\longrightarrow} 1.$  We highlight two examples. For the first we have convergence.

\begin{example}
	Let's consider $\rho_t=\frac{t}{t+1}$, then $\prod_{i=1}^{t} \rho_i=\frac{1}{t+1}$ which converges to 0 as $t \to \infty$.
\end{example}  
However, condition \eqref{eq:BM2} also ensures we don't have the following situation.
\begin{example}
	Let's consider $\rho_t=\exp(-\frac{1}{t^2})$, then $\prod_{i=1}^{t} \rho_i=\exp(-\sum_{k=1}^{t}\frac{1}{k^2})$ which does not converge to zero.
\end{example}  
Condition \eqref{eq:BM2} can also be written as 
$
\sum_{i=1}^{t} \log \rho_i \xrightarrow[t\rightarrow\infty]{} -\infty,
$ or differently as
$
\sum_{i=1}^{t} (-\log \rho_i )\xrightarrow[t\rightarrow\infty]{} \infty$.  In fact, this is the case if $\frac{-\log \rho_t}{t^{-\alpha}}\geq C$ for some $C>0$ and $\alpha>0$. This translates into
$\rho_t \leq e^{-Ct^{\alpha}}$.

An astute reader would immediately notice that $\rho_t=| A|_{\infty}$ is one of the possible norms one can use.  In fact we can replace the matrix norm with any matricial norm generated by a vector norm in $\R^n$.  More precisely we can take $\rho_t=| A|_{\alpha}$ where $|A|_\alpha=\max_{| x|_\alpha=1}|Ax|_\alpha$ with $|\cdot|_\alpha$ chosen to be a vector norm on $\R^n$.  For details on matricial norms, one can take a look at \cite{horn2012matrix}.  
However the infinity norm has a clean and clear interpretation in terms of $a_{ii}$ and $\epsilon_i$ and that is the main reason we use it throughout the paper.

We can extend the conclusions if we replace the $\infty$-norm  of a vector by something of the form
\[
|\nu|_{\infty,\beta}=\max_{i=1,\cdots,n}|\nu_i|/\beta_i
\] where $\beta$ is a vector of positive values such that $A\beta \leq \delta \beta$. In this new norm we now have
\begin{align*}
	|(Y_t)_i| &\leq \sum_{j=1}^{n} |b_{ij}| |(Y_{t-1})_j|\\
	\dfrac{|(Y_t)_i|}{\beta_i}&\leq \sum_{j=1}^{n} \dfrac{|b_{ij}|\beta_j}{\beta_i} \dfrac{|(Y_{t-1})|_j}{\beta_j},
\end{align*} which yields \begin{align*}
	|(Y_t)|_{\infty,\beta} &\leq |(Y_{t-1})|_{\infty,\beta}\max_{i=1,\cdots,n}\sum_{j=1}^{n} \dfrac{|b_{ij}|\beta_j}{\beta_i} \\
	&= |(Y_{t-1})|_{\infty,\beta}\max_{i=1,\cdots,n}\left( |a_{ii}-\eps_i|  + \dfrac{1}{\beta_i} \sum_{j\neq i}^{n} a_{ij}\beta_j\right).
\end{align*} 

From the assumption $A\beta \leq \delta \beta$ we can get in the first place that $\sum_{j=1}^{n} a_{ij}\beta_j \leq \delta\beta_i $ or $\sum_{j\neq i}^{n} a_{ij}\beta_j \leq \beta_i (\delta -a_{ii})$ and thus $\frac{1}{\beta_i}\sum_{j\neq i}^{n} a_{ij}\beta_j \leq (\delta -a_{ii})$. This yields
\[
|Y_t|_{\infty,\beta} \leq |(Y_{t-1})|_{\infty,\beta}\max_{i=1,\cdots,n} \left( |a_{ii}-\eps_i|  + \delta -a_{ii} \right)
\] as long as $|a_{ii}-\eps_i| +\delta -a_{ii} < 1$, which is satisfied by $-(1-\delta)<\eps_i<1-\delta+2a_{ii}$. The question is if there exists such a vector with $A\beta \leq \delta \beta $ (this means component wise). Such a choice is $\beta=[1,1,\dots,1]'$  and $\delta=1$ since $A$ is a stochastic matrix. If such a $\beta$ exists with $\delta <1$ then we get a relaxation of the main condition.

Interestingly, if $A$ is not necessarily stochastic but has positive entries, then by a theorem of Perron-Frobenius there exists a real eigenvalue that is greater than the absolute value of all the other eigenvalues and its eigenvector has positive entries. The argument above shows that we can definitely choose $\delta$ and $\beta$ to have the same result.

The above arguments allow us to posit this result.
\begin{theorem}
	Assume $X_t=A_tX_{t-1} + \mathcal{E}_t(\ms -X_{t-1})$ with $A_t$ row-stochastic matrix and let $\rho_t=\max_{i=1,\cdots,n}(|(a_t)_{ii}-(\eps_t)_i|+1 -(a_t)_{ii})$.  
	If  \, $\prod_{s=1}^{t}\rho_s  \underset{t \to \infty}{\longrightarrow} 0$, then $X_t \underset{t \to \infty}{\longrightarrow} \ms$.
\end{theorem}
In the case $A_t$ are all equal to A, then if $0<\eps_i<2a_{ii}, \: i=1,\cdots,n,$  then $X_t \underset{t \to \infty}{\longrightarrow} \ms$.

\section{Learning with  random noise}\label{s:LearningwithRandomnoise}
Our base model with learning is expanded to have random noise in the feedback term. We introduce a random vector $\gamma_t$ which we quantify later. The hypothesis is that $\gamma_t$ is small.  For this section we also consider the case of time depending evolution. 

The model is given by 
\[
X_t=A_tX_{t-1} + \mathcal{E}_t(\ms +\gamma_t-X_{t-1})
\] where $X_t$ is the vector of prices at time $t$ and $\ms$ is the vector of equilibrium price or consensus value the agents are trying to learn. In order to prove that $X_t - \ms$ converges to 0, we rewrite the equation as 

\begin{align*}
	X_t - \ms &=A_t X_{t-1} -\ms + \mathcal{E}_t(\ms -X_{t-1}) +\mathcal{E}_t \gamma_t\\
	&=A X_{t-1} - A_t\ms + \mathcal{E}_t(\ms -X_{t-1}) +\mathcal{E}_t \gamma_t \mbox{ as $A\ms=\ms$ }\\
	&=(A_t  -\mathcal{E}_t)(X_{t-1} -\ms) + \mathcal{E}_t \gamma_t .\\
\end{align*}
Therefore if we denote by $Y_t=X_{t-1} -\ms$, then we can simplify the above expression as
\[
Y_t=(A_t  -\mathcal{E}_t)Y_{t-1} + \mathcal{E}_t \gamma_t .
\]
With the same argument as before we obtain
\[
|Y_t|_\infty \leq \rho_t |Y_{t-1}|_\infty + C|\gamma_t|
\]
with 
\begin{equation}\label{e:l2}
	\rho_t=\max_{i=1,\cdots,n} (|(a_t)_{ii}-(\eps_t)_i|+1-(a_t)_{ii}).
\end{equation}
We formulate a general result as follows.   
\subsection{Noisy Learning}
In the theorem below, we examine the appropriate noise in convergence terms.  With vanishing noise, the system still exhibits the consensus property. Proving this convergence with vanishing noise in probability requires a separate lemma. Noise has two parts in this theorem. In the first three parts, noise is vanishing and consensus is reached. In markets, this noise is seen as a form of shock that decays over time. The Theorem quantifies this precisely, using distinct modes of convergence.
In the fourth item, when the noise is persistent agents do not reach consensus and also do not converge to the same asymptotic distributions. So while $X_t$ converges asymptotically, the individual agents may converge to different marginal distributions.

\begin{theorem}\label{thm:blwn}
	Assume the model $X_t=A_tX_{t-1} + \mathcal{E}_t(\ms +\gamma_t-X_{t-1})$ with $A_t$ a row-stochastic matrix. With the notation from \eqref{e:l2} assume that 
	\begin{equation}\label{e:ll1} 
 \sup_{t\ge1}\{\rho_t+\rho_t\rho_{t-1}+\rho_t\rho_{t-1}\rho_{t-2}+\dots+\rho_{t}\rho_{t-1}\dots\rho_1\}<\infty.
	\end{equation}
	
	\begin{enumerate}
		\item If $\gamma_t \xrightarrow[t \to \infty]{a.s} 0$, then $X_t \xrightarrow[t \to \infty]{a.s} \ms$.
		\item  If $\gamma_t \xrightarrow[t \to \infty]{\p} 0$, then $X_t \xrightarrow[t \to \infty]{\p} \ms$.
		\item  If $\gamma_t \xrightarrow[t \to \infty]{L^p} 0$, then $X_t \xrightarrow[t \to \infty]{L^p} \ms$.
		\item If we assume 
		\begin{equation}\label{e:llc14}
			X_t=A_tX_{t-1} + \mathcal{E}_t(\gamma_t-X_{t-1})
		\end{equation}
		where now $(\gamma_t)_{t\ge1}$ are iid and integrable and in addition to \eqref{e:ll1} we assume that 
		\begin{equation}\label{e:ll1b}
			\sum_{t\ge1}(|A_t-A_{t-1}|_\infty+|\mathcal{E}_t-\mathcal{E}_{t-1})|_{\infty})<\infty.
		\end{equation}
		Then,  there exists an $n$-dimensional random variable $X$ such that 
		\begin{equation}\label{e:llc11}
			X_t \Longrightarrow X
		\end{equation}
  where the convergence is distribution as pointed in \eqref{e:convdist}.
		\item Furthermore, if $\gamma_t$ is integrable but not constant almost surely, then, without condition \eqref{e:ll1b}, the conclusion of \eqref{e:llc11} does not hold.
	\end{enumerate}
\end{theorem}
\begin{proof}
\begin{enumerate}
\item  From our base model in terms of $Y_t$ is 
\begin{equation}\label{e:ll6}
	Y_t=(A_t  -\mathcal{E}_t)Y_{t-1} + \mathcal{E}_t \gamma_t .
\end{equation}	
From this we get 
\begin{equation}\label{e:ll4}
	|Y_t|_\infty \leq \rho_t |Y_{t-1}|_\infty + C|\gamma_t|_{\infty}.
\end{equation}
If we assume that $|\gamma_t|_{\infty} \xrightarrow[t \to \infty]{a.s} 0$,  then we get that $|Y_t|\xrightarrow[t \to \infty]{a.s}0$.  Indeed, this becomes a purely deterministic statement. For a given $\epsilon >0$, we can find that $|\gamma_t|_{\infty} \leq \epsilon$ for all $t \geq t_{\epsilon}$. Then,
\[
|Y_t|_\infty \leq \rho_t |Y_{t-1}|_\infty + C\epsilon \quad \forall t \geq t_{\epsilon}.
\] 
Using the previous inequalities for $t-1$, $t-2$, \dots, $t_\epsilon$  gives that 
\[
|Y_t|_\infty \leq (\prod_{s=t_\epsilon}^t \rho_{s}) |Y_{t_\epsilon -1}|_\infty + C\epsilon (1+ \rho_t  + \cdots+ \prod_{s=t_\epsilon}^t \rho_{s}).
\] 
From \eqref{e:ll1} combined with the following elementary lemma we show that $|Y_t|_\infty \xrightarrow[t\to\infty]{a.s.}0$.  

\begin{lemma}\label{L:2b} Assume that $\{\rho_t\}_{t\ge1}$ is a sequence of non-negative numbers such that for some $A>0$,  and any $t\ge1$,  
	\begin{equation}\label{e:ll13}
		1+\rho_t+\rho_{t}\rho_{t-1}+\dots+\rho_{t}\rho_{t-1}\dots \rho_1\le A. 
	\end{equation}
	Then, for $0\le s\le t-1$, 
	\begin{equation}\label{eL:2:2}
		\rho_{t}\rho_{t-1}\dots\rho_{s+1}\le A e^{-c(t-s)}, \text{ where }c=\ln\left (1+1/A\right),
	\end{equation}
	and in addition, 
	\begin{equation}\label{e:llc3}
		\rho_t\rho_{t-1}\dots \rho_{t-s+1}(1+\rho_{t-s}+\rho_{t-s}\rho_{t-s-1}+\dots +\rho_{t-s}\dots\rho_{1})\le A^2 e^{-c s}. 
	\end{equation}
	
	It is also true that \eqref{eL:2:2} for some constants $c>0$ and $A>0$ implies \eqref{e:ll13} with the bound on the right being $A/(e^c-1)$.    
\end{lemma}

\begin{proof}  To see this we first denote 
\[
A_t=\rho_t+\rho_{t}\rho_{t-1}+\dots+\rho_{t}\rho_{t-1}\dots\rho_{0}.
\]
Then we get that 
\[
\rho_{t}=\frac{A_t}{1+A_{t-1}}
\]
and thus 
\[
\begin{split}
	&\rho_{t}\rho_{t-1}\dots\rho_{s+1}= \frac{A_t}{1+A_{t-1}}\frac{A_{t-1}}{1+A_{t-2}}\dots\frac{A_{s+1}}{1+A_{s}} \\ 
	&=\frac{A_t}{1+A_{s}}\left(1-\frac{1}{1+A_{t-1}} \right)\left(1-\frac{1}{1+ A_{t-2}} \right)\cdots \left(1-\frac{1}{1+ A_{s+1}} \right) \\
	&\le A\left(1-\frac{1}{1+A}\right)^{t-s}=Ae^{-c(t-s)}.  
\end{split}
\]
To see \eqref{e:llc3}, we only need to notice that 
\[
\rho_t\rho_{t-1}\dots \rho_{t-s+1}(1+\rho_{t-s}+\rho_{t-s-1}\rho_{t-s-2}+\dots +\rho_{t-s}\dots\rho_{1})\le A^2e^{-cs}.  
\]

It is a simple exercise to go from \eqref{eL:2:2} back to \eqref{e:ll13}.  
\end{proof}

\item 
If we only assume a weaker condition, namely that $\gamma_t\xrightarrow[t \to \infty]{\p}0$ (only convergence in probability), then iterating \eqref{e:ll4} we obtain 
\begin{equation}\label{e:ll5}
|Y_t|_{\infty}\le (\prod_{s=1}^t\rho_s)|Y_0|_{\infty}+\sum_{s=0}^{t}(\prod_{i=t-s+1}^t\rho_i)|\gamma_{t-s}|_{\infty}
\end{equation}
with the convention that $\prod_{i=t+1}^t\rho_i=1$.  

To finish the proof off we use the following Lemma with $u_t=|\gamma_t|_{\infty}$.

\begin{lemma}\label{L:2}
	Let $(u_n)_{n\geq 1}$ be a random sequence such that
	\begin{align}
		&\label{e:llc1} u_n \xrightarrow[n \to \infty]{\p}0 \\
		&\label{e:llc2} \p(\sup_{n \geq 1}|u_n|< \infty)=1.
	\end{align}
	Then, under the assumption \eqref{e:ll1}, we have the convergence $\sum_{i=1}^{t}\rho_t\rho_{t-1}\dots \rho_{t-i+1}u_{t-i}\xrightarrow[t \to \infty]{\p} 0.$
\end{lemma}

\begin{proof}
	
	For the argument, denote for simplicity of writing $\eta_{t,i}=\rho_t\rho_{t-i}\dots \rho_{t-i+1}$.  
	
	Now, we fix $s\le t$ and write  
	\[
	|\sum_{i=1}^{t} \eta_{t,i} u_{t-i}| 
	\leq \sum_{i=1}^{s-1} \eta_{t,i} |u_{t-i}| + \sum_{i=s}^{t} \eta_{t,i}|u_{t-i}|
	\]
	Now, for a given $\epsilon$ and  $|\sum_{i=1}^{t} \eta_{t,i} u_{t-i}|>\eps$, we must have that at least one of the above sums must be at least $\eps/2$, thus,  we can write for each fixed $\epsilon>0$, 
	\begin{equation}\label{e:1}
		\p(|\sum_{i=1}^{t} \eta_{t,i} u_{t-i}|>\eps)\le \p(\sum_{i=1}^{s-1} \eta_{t,i} |u_{t-i}|\ge\eps/2)+\p(\sum_{i=s}^{t} \eta_{t,i} |u_{t-i}|>\eps/2).
	\end{equation}
	The next step is to use the boundedness of $u_t$.   Take arbitrary 
	$\delta,M>0$, (here $\delta$ is meant to be small and $M$ to be large) and then set 
	\[
	A_{M}=\{|u_n|\le M\text{ for all }n\ge1 \}.
	\]
	From the condition \eqref{e:llc2} we definitely have that $\p(A_M)$ converges to $1$ as $M$ tends to infinity. 
	Therefore we can continue the equation \eqref{e:1} with 
	\begin{align*}
		\p(|\sum_{i=1}^{t} \eta_{t,i} u_{t-i}|>\eps)&\le\p(\sum_{i=1}^{s-1}\eta_{t,i} |u_{t-i}|\ge \eps/2) +\p(\sum_{i=s}^{t} \eta_{t,i} |u_{t-i}|>\eps/2,A_M) + \p(\sum_{i=s}^{t} \eta_{t,i}|u_{t-i}|>\eps/2,A_M^c) \\
		& \le \sum_{i=1}^{s-1}\p(\eta_{t,i} |u_{t-i}|\ge \eps/(2(s-1)))
		+\p(M\sum_{i=s}^{t} \eta_{t,i}>\eps/2,A_M)+\p(A_M^c) \\
		& \le \sum_{i=1}^{s-1}\p(\eta_{t,i} |u_{t-i}|\ge \eps/(2(s-1)))
		+\p(\sum_{i=s}^{t} \eta_{t,i}>\eps/(2M))+\p(A_M^c)\\ 
		& \le \sum_{i=1}^{s-1}\p(\eta_{t,i} |u_{t-i}|\ge \eps/(2(s-1)))
		+\p(A^2e^{-cs}>\eps/(2M))+\p(A_M^c)
	\end{align*}
	where in the passage from the first line to the second we 
	used the union bound, more precisely, if we have 
	$\sum_{i=1}^s\eta_{t,i} |u_{t-i}|\ge \eps/2$ then at least one of the 
	terms must be $\ge \eps/(2s)$ plus the union bound on the probability.   Finally in passage to the last line we simply used \eqref{e:llc3}.  
	
	Next we can freeze for now $\eps,s,M$ and use the fact that for each $i$, 
	$\eta_{t,i} u_{t-i}$ converges to $0$ in probability since $\eta_{t,i}$ is bounded by $A>0$ and use \eqref{e:llc3}  to argue that the limit as $t\to\infty$ we gain that 
	\[
	0\le \limsup_{t\to\infty}\p(|\sum_{i=1}^{t} \eta_i u_{t-i}|>\eps)\le 
	\p(A^2 e^{-cs}>\eps/(2M))+\p(A_M^c).
	\]
	For large $s$, obviously $\p(A^2e^{-cs}>\eps/(2M))=0$ and thus we arrive at   
	\[
	0\le \limsup_{t\to\infty}\p(|\sum_{i=1}^{t} \eta_i u_{t-i}|>\eps)\le \p(A_M^c).
	\]
	From this, we take the limit as $M\to\infty$ and using \eqref{e:llc1}
	\[
	0\le \limsup_{t\to\infty}\p(|\sum_{i=1}^{t} \eta_i u_{t-i}|>\eps)=0
	\]
	which means convergence of $\sum_{i=1}^{t} \eta_i u_{t-i}$ to $0$ in probability. 
\end{proof}

Now let's return to the proof of the Theorem.  

\item For the $L^p$ convergence we just need to take expectation of \eqref{e:ll5}. 

\item 
For the convergence in distribution we start by writing 
\[
X_t=B_tX_{t-1}+\mathcal{E}_t\gamma_t
\]
where $B_t=A_t-\mathcal{E}_t$.  The idea is that because $\gamma_t$ are in $L^1$ so are all the variables $X_t$.  We are going to use the Wasserstein distance to control the difference between the distributions of $X_t$ and $X_{t-1}$. 

The basic idea is that in a slightly modified Wasserstein distance $D$ we have a contraction in the sense that there exists some $\rho<1$ such that  
\begin{equation}\label{e:n3}
	D(X_t,X_{t-1})\le \rho D(X_{t-1},X_{t-2}). 
\end{equation}
For the sake of completeness we define here for two $n$-dimensional random variables, $X,Y$ or better for their distributions $\mu_X,\mu_Y$, 
\begin{equation}\label{e:D}
	D(X,Y)=\left(\inf_{\alpha}\int  |x-y|_\infty\alpha(dx,dy) \right)=\inf_{\alpha}\E[|\tilde{X}-\tilde{Y}|_\infty]
\end{equation}
where $\alpha$ is a $2n$-dimensional distribution with marginals $\mu_X$ and $\mu_Y$ and $\tilde{X}$ $\tilde{Y}$ are two random variables on the same probability space (we call it a coupling) with the same distributions as $X$, respectively $Y$.  The second equality follows easily from taking $\tilde{X}$ and $\tilde{Y}$ to be the projections from $\pi_{i}:\R^n\times\R^n\to\R^n$, given by $\pi_1(x,y)=x$ while $\pi_2(x,y)=y$.  To go from the pair $(\tilde{X},\tilde{Y})$ back to the measure $\alpha$, we just need to take $\alpha$ to be the distribution of the pair $(\tilde{X},\tilde{Y})$.    

The standard Wasserstein distance is defined as 
\[
W_1(X,Y)=\left(\inf_{\alpha}\int  |x-y|\alpha(dx,dy) \right)=\inf\E[|\tilde{X}-\tilde{Y}|]. 
\]

Because any two norms on $\R^n$ are equivalent, we can find two constants $c_1,c_2>0$ such that 
\[
c_1 W_1(X,Y)\le  D(X,Y)\le c_2W_1(X,Y).  
\]
It is known that $W_1$ gives the topology of weak converge on the space of probability measures with finite first moment (that is $\int |x|\mu(dx)<\infty$).  Due to the above inequality we also infer the completeness with respect to the metric $D$ on the same space $\mathcal{P}_1(\R^n)$.  

To carry on this program we define for a distribution $\mu$, the following map
\[
F_t(\mu)=\text{ the distribution of } g_t(X_{t-1},\gamma) \text{ with } g_t(x,\lambda)=(A_t-\mathcal{E}_t)x+\mathcal{E}_t\lambda, x,\lambda,\in\R^n,
\]
where $X$ is a random variable with distribution $\mu$ and $\gamma$ is a random variable independent of $X$ and having the same distribution as the sequence $\gamma_t$. 

Now we want to look at $D(X_t,X_{t-1})$ and estimate it from above.  To do this assume that we have a coupling between $X_{t-1}$ and $X_{t-2}$ and then we can create an optimal coupling between $X_t$ and $X_{t-1}$ (with respect to the distance $D$, which certainly exists from Kantorovich general result) and then take $\gamma$ independent of both $X_{t-1}$ and $X_{t-2}$ and use 
\begin{align*}
	X_t-X_{t-1}&=(A_t-\mathcal{E}_t)X_{t-1}+\mathcal{E}_t\gamma-(A_{t-1}-\mathcal{E}_{t-1})X_{t-2}-\mathcal{E}_{t-1}\gamma \\ 
	&=(A_t-\mathcal{E}_t)(X_{t-1}-X_{t-2})+(A_t-A_{t-1}-\mathcal{E}_t+\mathcal{E}_{t-1})X_{t-2}+(\mathcal{E}_t-\mathcal{E}_{t-1})\gamma.
\end{align*}
Taking $|\cdot|_{\infty}$ and the expectation both sides we get the estimate 
\begin{equation}\label{e:ll6b}
	\begin{split}
		\E[|X_t-X_{t-1}|_{\infty}]&\le \E[|(A_t-\mathcal{E}_t)X_{t-1}|_\infty]+\E[|(A_t-A_{t-1}-\mathcal{E}_t+\mathcal{E}_{t-1})X_{t-2}|_1] \\ \qquad &+\E[|(\mathcal{E}_t-\mathcal{E}_{t-1})\gamma|_\infty]\\
		& \le \rho_t\E[|X_{t-1}-X_{t-2}|]+\alpha_{t}(\E[|X_{t-2}|_\infty]+\E[|\gamma|])
	\end{split}
\end{equation}
where we denoted by 
\[
\alpha_t=|A_t-A_{t-1}|_\infty+|\mathcal{E}_t-\mathcal{E}_{t-1}|_\infty. 
\]

Notice that in the time independent case, the terms $\alpha_t$ is  0, which  implies that $X_t$ converges in distribution.   

In the general case we need to use the extra conditions from \eqref{e:ll1b}.  From the above considerations we actually show first that the expectation of $X_t$ obeys the equation (keep in mind that $\sup_{t\ge1}|\mathcal{E}_t|_\infty\le A+1$)  
\[
\E[|X_t|_\infty]\le \rho_{t}\E[|X_{t-1}|_\infty]+(A+1)\E[|\gamma|_\infty].  
\]
Using this and the standard iterations combined with \eqref{e:ll1} we get that 
\[
\sup_{t}\E[|X_t|_\infty]<C<\infty. 
\]
On the other hand from \eqref{e:ll6} we get that 
\begin{equation}\label{e:ll12}
	D(X_t,X_{t-1})\le \rho_t D(X_{t-1},X_{t-2})+C\alpha_t.  
\end{equation}
Using this and a simple iteration it leads to 
\[
D(X_t,X_{t-1})\le \rho_t \rho_{t-1}\dots\rho_2 D(X_1,X_0)+C(\alpha_t +\alpha_{t-1}\rho_t+\alpha_{t-1}\rho_{t}\rho_{t-1}+\dots+\alpha_{1}\rho_t\rho_{t-1}\dots\rho_{1}).  
\]
In particular, summing this over $t$ from $t$ to $t+s$, leads to 
\[
D(X_t,X_{t+s})\le \sum_{i=1}^s\rho_{t+i-1}\dots \rho_2 D(X_1,X_0)+C\sum_{k=1}^{t+s}\alpha_k \sum_{i=1}^{s}\rho_{t+i}\rho_{t+i-1}\dots \rho_{k}.
\]
According to \eqref{e:llc3} we conclude that the sum $\sum_{i=1}^s\rho_{t+i-1}\dots \rho_2$ converges to $0$ as $s,t\to\infty$.  We will show that the other sum also converges to $0$ as both $t,s\to\infty$.  To this end notice that from \eqref{e:ll1b}, we can set 
\[
\beta_t=\sum_{i\ge t}\alpha_i.
\]
and write $\alpha_t=\beta_t-\beta_{t+1}$.  After rearrangements, this leads to   
\[
\sum_{k=1}^{t+s}\alpha_k \sum_{i=1}^{s}\rho_{t+i}\rho_{t+i-1}\dots \rho_{k}=\beta_1\rho_{t+s}\rho_{t+s-1}\dots \rho_{1} +\beta_2\rho_{t+s}\rho_{t+s-1}\dots \rho_{1}+\dots +\beta_{t+s}.
\]

The first term converges to 0 because of \eqref{eL:2:2} and the rest, converges to $0$ because of Lemma~\ref{L:2} thanks to the fact that $\beta_t$ converges to $0$, this converges to $0$. 

This proves the convergence in distribution.   

\item Next we show that the condition \eqref{e:ll1b} is also a necessary condition.  Indeed, if we take the one dimensional case with 
\[
X_t=X_{t-1}+\epsilon_t(\gamma_t-X_{t-1})
\]
such that 
\[
|\epsilon_t-\epsilon_{t-1}|=1/(10t) \text{ for }t\ge1  
\]
In fact we will choose 
\[
\epsilon_{t}=1/2+c\sum_{k=1}^t w_{i}/i
\]
and we will choose $w_i=\pm1$ in the following fashion.  First we take all $w_1,w_2,\dots, w_{\tau_1}$ such that $\epsilon_{\tau_1}\le 3/4$ but $3/4<\epsilon_{\tau_1}+c/(\tau_{1}+1)$.  Notice that we can do this because the harmonic series is divergent.  Now, we choose $\tau_2>\tau_1$ such that $w_{\tau_1+1}=w_{\tau_1+2}=\dots=w_{\tau_2}=-1 $ and $\epsilon_{\tau_2}-1/(10(\tau_2+1))<1/4\le \epsilon_{\tau_2}$.  Now we choose $\tau_3>\tau_2$ and $w_{t_2+1}=\dots= w_{t_3}=1$ such that $\epsilon_{\tau_3}\le 3/4<\epsilon_{\tau_3}+c/(\tau_{3}+1)$.  Then we choose $\tau_4>\tau_3$ such that $w_{\tau_3+1}=w_{\tau_3+2}=\dots=w_{\tau_4}=-1 $ such that $\epsilon_{\tau_4}-1/(10(\tau_4+1))<1/4\le \epsilon_{\tau_4}$.  And we continue inductively.  Thus we have defined a sequence $\epsilon_t$ such that 
\[
1/4\le \epsilon_t\le 3/4 \text{ such that } \overline{\{\epsilon_t\}_{t\ge1}}=[1/4,3/4].      
\]
In other words the limit points of the sequence $\epsilon_t$ is just the interval $[1/4,3/4]$ and obviously the condition \eqref{e:ll1} is fulfilled.  

With this choice of the sequence $\epsilon_t$, we claim that the sequence $X_t$ does not converge in distribution.  Indeed the argument is based on the simple observation that if it were, then taking the characteristic functions $\phi_{X_t}$ we would get 
\[
\phi_{X_t}(\xi)=\phi_{X_{t-1}}((1-\epsilon_t)\xi)\phi_{\gamma}(\epsilon_t\xi).  
\]
As a recall, $\phi_X(\xi)=\E[e^{i\xi X}]$ for any $\xi\in\R$.  
In particular this means that if $X_t$ converges to some random variable $Y$, then taking a subsequence $t_n$ for which $\epsilon_{t_n}\xrightarrow[n\to\infty]{}x$ we obtain that 
\begin{equation}\label{e:ll11}
	\phi_Y(\xi)=\phi_{Y}((1-x)\xi)\phi_{\gamma}(x\xi) \text{ for any }x\in[1/4,3/4].  
\end{equation}
Under the assumption that $\gamma$ is integrable we claim that $\gamma$ must be constant and also $X$ is going to be the same constant.  To carry this out we argue that for $x=1/4$ and $x=3/4$ we get that 
\[
\frac{\phi_{\gamma}(3\xi/4)}{\phi_{\gamma}(\xi/4)}=\frac{\phi_{Y}(3\xi/4)}{\phi_Y(\xi/4)}. 
\]
Replacing $\xi$ by $4\xi/3$ we arrive at 
\[
\frac{\phi_\gamma(\xi)}{\phi_\gamma(\xi/3)}=\frac{\phi_Y(\xi)}{\phi_Y(\xi/3)}. 
\]
Replacing here $\xi$ by $\xi/3$, $\xi/3^2$, \dots, $\xi/3^n$ and multiplying these we obtain
\[
\frac{\phi_\gamma(\xi)}{\phi_\gamma(\xi/3^n)}=\frac{\phi_Y(\xi)}{\phi_Y(\xi/3^n)}. 
\]
Now letting $n\to\infty$ and using the fact that for any random variable $Z$,  $\phi_Z(\xi/3^n)\xrightarrow[n\to\infty]{}1$ we obtain that 
\[
\phi_\gamma(\xi)=\phi_Y(\xi), 
\]
in other words, $Y$ has the same distribution as $\gamma$.   Using this in \eqref{e:ll11} with $x=1/2$ we arrive at 
\[
\phi_Y(\xi)=\phi_{Y}(\xi/2)^2.
\]
Iterating this we get 
\[
\phi_Y(\xi)=\phi_{Y}(\xi/2^n)^{2^n}
\]
which can be written alternatively as 
\begin{equation}\label{e:llc12}
	\phi_Y(\xi)=\phi_{\frac{Y_1+Y_2+\dots+Y_{2^n}}{2^n}}(\xi),
\end{equation}
where $Y_1,Y_2,\dots$ are iid with the same distribution as $Y$.  Since $Y$ and $\gamma$ have the same distributions and $\gamma$ is integrable, it follows that $Y$ is also integrable.  This in particular implies from the law of large numbers that $\frac{Y_1+Y_2+\dots+Y_{2^n}}{2^n}$ converges almost surely to $\E[Y]=\E[\gamma]$.  Since convergence almost surely implies convergence in distribution, we get that  
\begin{equation}\label{e:llc13}
	\phi_{Y}(\xi)=\phi_{\E[Y]}(\xi), 
\end{equation}
in other words, $Y$ must be constant.  This implies that $\gamma$ is also constant which then finishes the argument.  
\end{enumerate}
\end{proof}

As the proof is quite lengthy and requires two technical lemmas, we defer it to the appendix. Observe here the fact that in the last part of the Theorem we incorporated the constant $\ms$ into $\gamma_t$.  The convergence is in distribution sense and thus it does not lead to convergence as in the previous cases.  Even if we assume that $\gamma_t$ is of the form $\bar{\sigma}+\gamma_t$, the convergence will not be to $\ms$ alone.  Thus this is a different convergence scenario and in spirit is not of the same form as the other cases.   
\begin{remark}
We need to point out that integrability is key for the conclusion of the last part of Theorem~\ref{thm:blwn}. This shows the intricate relationship \eqref{e:ll1b} and \eqref{e:ll11}.  If we drop the integrability and for instance take $(\gamma_t)_{t\ge1}$ to be all iid Cauchy($1$) and $X_0=0$, then $X_t$ will also follow a Cauchy($1$) random variable for any choice of $0\le\epsilon_{t}\le 1$ with $\epsilon_1>0$.  Certainly, in this situation we do not need any other assumptions on $\epsilon$ or $\rho_t$ or condition \eqref{e:ll1b} to get convergence of $X_t$. We leave as an open problem the optimal conditions under which the model \eqref{e:llc14} converges as $t\to\infty$.  
\end{remark}

Condition \eqref{e:ll1b} elucidates a key behavioural aspect. Agents are comfortable and so stabilize their trust matrix $A$ and learning ability $\mathcal{E}$. Like members of a small village or an island, everyone knows over time how much trust to place on themselves and the other members.

\subsection{Simulations for convergence to distribution}
Let us illustrate Theorem \ref{thm:blwn} and result \ref{e:llc11}. Suppose that the noise $\gamma_{t}$ is a Normal random variable. Numerical simulations show that $X_t$ converges to a Gaussian random variable for each component: figure \ref{f:nonCLT1}. 
\begin{figure}%
\centering
\textbf{Joint plot of contours with marginals}\\
\subfloat[\centering Gaussian Noise]{{\includegraphics[width=6cm]{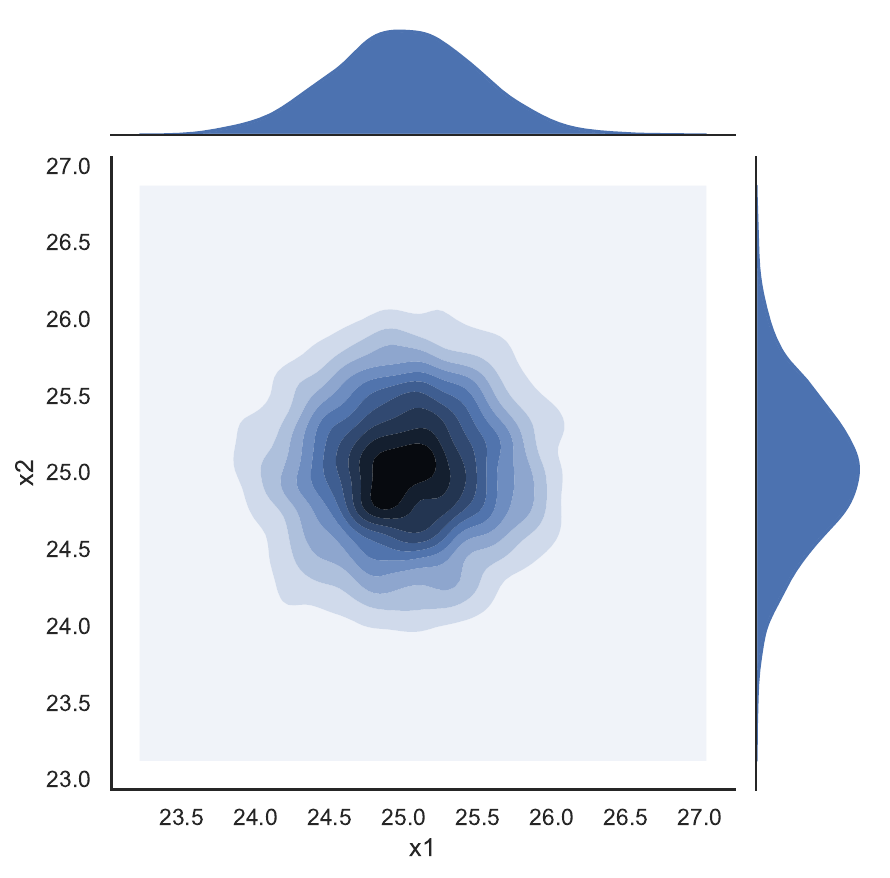} }}%
\qquad
\subfloat[\centering Non-Gaussian Noise]{{\includegraphics[width=6cm]{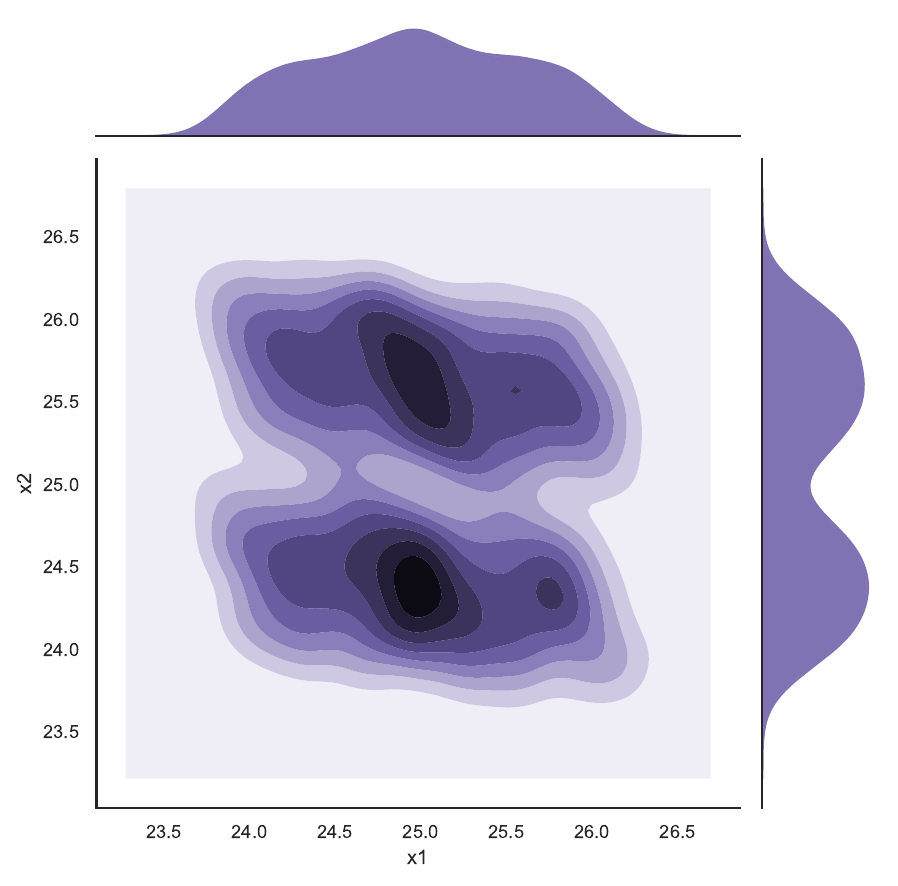} }}%
\caption{Plots illustrates the case for two agents who learn from each other with $A$ and $\mathcal{E}$ fixed. Variable x1 and x2 represent agents 1 and 2: (a) With Gaussian noise, $X_t$ converges to a Normal distribution. (b) Non-Gaussian noise generates a different asymptotic distribution.}\label{f:nonCLT1}%
\end{figure}

The the asymptotic distribution is Gaussian centered around the true value $\ms$. The key point is that we do not need to scale $X_t$. Suppose, the iid $(\gamma_{t})s$ are vectors of just $+1$ or $-1$, then $X_t$ converges in distribution. In figure \ref{f:nonCLT1} (b), the simulated distribution looks distinctly non-Gaussian. For other noises, different distributions can occur. Note that even if agents converge in distribution, their marginal distributions can be distinct. 

\section{Nonlinear learning}\label{s:nl} \label{s:Nonlinearlearning}
While Averaging (updating) is retained in this section, we develop nonlinear models of learning. Instead of $\mathcal{E}$, there is a non-linear function.

\begin{definition}
	The learning function is $f_t: \R^n \to \R^n$ is  continuous on some compact convex subset $K \subseteq \R^n$ and differentiable on its interior, with $f(0)=0$. Component wise it is
	\[
	f_t \left( \begin{pmatrix}
		x_1\\
		\vdots\\
		x_n
	\end{pmatrix} \right)=\begin{pmatrix}
		f_t(x_1)\\
		\vdots\\
		f_t(x_n)
	\end{pmatrix}.
	\]
\end{definition}

Notice that the update or feedback is now varying with time.  Learning or feed back stops when $\ms-X_t=0$, so the condition $f_t(0)=0$ ensures this.  
The updating rule for agent $i$ becomes 
\[
x^i_{t+1}=\sum_{j=1}^n (a_t)_{ij} x^j_{t}+f_{t,i}(\ms-x^i_{t}). 
\]
Moreover, the weights matrix $A_t$ is also time varying and row-stochastic. Previous sections showed convergence results of linear updating $f_{t,i}=\epsilon_i$, a fixed scalar. Actual updating of feedback can be be quite complex, and having a nonlinear feedback or learning rule allows us to expand the linear model. 

\begin{theorem} \label{thm:noisenonlinear}
	For $ \forall i \in \lbrace 1,\cdots,n \rbrace \mbox{ and } \forall t \geq 0$, suppose the learning function satisfies
	\begin{equation}\label{e:n1}
		0< \inf f_{t,i}^\prime \leq \sup f_{t,i}^\prime < 2(a_t)_{ii}, 
	\end{equation}
	and if we denote 
	\[
	\rho_t=\sup_i \sup_\xi \left(  |(a_t)_{ii}-f_{t,i}'(\xi)|+1-(a_t)_{ii}\right)
	\]
	 and we assume that 
	\begin{equation}\label{e:n4}
		\sup_{t\ge1}(\rho_t+\rho_{t}\rho_{t-1}+\dots +\rho_{t}\rho_{t-1}\rho_{t-2}\dots \rho_{1})<\infty.
	\end{equation}
	Then,
	\begin{enumerate}
		\item With the dynamics
		\[
		X_{t}=A_tX_{t-1}+ f_t(\ms - X_{t-1}),
		\] 
		consensus is reached and $\lim_{t \to \infty} X_t = \ms$.
		
		\item If the evolution is given by  
		\[
		X_t=A_tX_{t-1} + f_t(\ms +\gamma_t-X_{t-1})
		\]
		under the same assumption as in \eqref{e:n1}, then $\gamma_t \xrightarrow[t \to \infty]{} 0$ yields that $X_t \xrightarrow[t \to \infty]{} \ms$. \textit{(If the noise converges to zero $a.s$, in probability or in ${L^1}$, then $X_t$ converges accordingly)}.\\
		
		\item Again assume \eqref{e:n4} and 
		\begin{equation}\label{e:n10}
			X_t=A_tX_{t-1} + f_t(\gamma_t-X_{t-1})
		\end{equation}
		where the sequence $(\gamma_t)_{t\ge1}$ is assumed to be iid and integrable.  If in addition we have that 
		\begin{equation}\label{e:n5}
			\sum_{t\ge1}\left(|A_t-A_{t-1}|_\infty+\max_{i}\sup_{\xi\in\R}|f_{t,i}'(\xi)-f_{t-1,i}'(\xi)|\right)<\infty,
		\end{equation}
		then $X_t$ converges in distribution as $t\to\infty$ as in the definition of equation \eqref{e:convdist}.  
	
	\end{enumerate}
	
\end{theorem}

\begin{proof}

\begin{enumerate}
\item First we subtract $\ms$ from both sides of the dynamics equation. As $A$ is stochastic, $A(t)\ms=\ms$, hence
\[
(X_{t+1} -\ms)=A(t)(X_t-\ms) + f_t(\ms - X_{t-1}).
\]
Second, we recast the equation using the infinity-norm 
\[
|X_{t+1} -\ms|_\infty=\sup_{i}|(X_{t+1} -\ms)_i|.
\]
For individual $i$, the updating rule becomes
\begin{align*}
	(X_{t+1} -\ms)_i &=\sum_{j=1}^n (A_t)_{ij}(X_{t-1}-\ms)_j+ f_t(\ms - (X_{t-1})_i)\\
	&=\left((a_t)_{ii} -\dfrac{ f_{t,i}(\ms - (X_{t-1})_i)}{(\ms - (X_{t-1})_i)}\right)(X_{t-1}-\ms)_i+ \sum_{j\neq i}^n (A_t)_{ij}(X_{t-1}-\ms)_j\\
	&\leq\left(  |(a_t)_{ii}-f_{t,i}^\prime(\xi_i)|  |X_{t-1} -\ms|_i \right) +\sum_{j\neq i}^n (A_t)_{ij} |X_{t-1} -\ms|_j\\
	&\leq\left(  |(a_t)_{ii}-f_{t,i}^\prime(\xi_i)|  + 1-(a_t)_{ii}\right) |X_{t-1} -\ms|_\infty\\
	&\leq \sup_i \sup_\xi \left(  |(a_t)_{ii}-f_{t,i}^\prime(\xi_i)|  + 1-(a_t)_{ii}\right) |X_{t-1} -\ms|_\infty
\end{align*}
The second equality follows because the learning function is continuous and differentiable hence
\[
f_{t,i}(x)-f_{t,i}(0)=(x-0)f_{t,i}^\prime(\xi)\implies \dfrac{f_{t,i}(x)}{x}=f_{t,i}^\prime(\xi).
\] 
for some $\xi_i \in (0,x)$ by the Mean value theorem.

By assumption
\[
0< \inf f_{t,i}^\prime \leq \sup f_{t,i}^\prime < 2(a_t)_{ii}
\]

but this is equivalent there being some $ 0<\delta_i <1$ such that $\forall \xi \in \R$

\begin{equation}\label{eq:MVT2}
	\delta_i< f_{t,i}'(\xi) < 2(a_t)_{ii} - \delta_i.
\end{equation}

The above condition gives us two cases to consider. In the first case, ignoring dependence on t, for all $i \in \{1,\cdots,n\}$ and $\xi$
\[
a_{ii}>f_i'(\xi) \, \text{ (case 1) in which case, }
|a_{ii}-f_i'(\xi)|+1 -(a_{ii}) =1-f_{t,i}'(\xi)<1-\delta_i
\]
In the second case,
\[
a_{ii}\le f_i'(\xi) \, \text{ (case 1) in which case, }
|a_{ii}-f_i'(\xi)|+1 -(a_{ii}) =1 + f_{t,i}'(\xi)-2a_{ii}<1-\delta_i.
\]
Thus we obtain that 
\[
\sup_i \sup_\xi \left(  |(a_t)_{ii}-f_{t,i}^\prime(\xi_i)|  + 1-(a_t)_{ii}\right)<1-\min_{i}\delta_i<1
\]
thus we have a contraction in $|X_t-\ms|_\infty$ and consequently, 
\[
\lim_{t \to \infty} X_t = \ms.
\]

\item 
The deviation equation from consensus is
\[
(X_{t+1} -\ms)=A(t)( X_t -\ms)+ f_t(\ms +\gamma_t - X_{t-1}).
\]
Essentially the same steps follow as the in the proof with no noise
\begin{align*}
	(X_{t+1} -\ms)_i &=\sum_{j=1}^n (A_t)_{ij}(X_{t-1}-\ms)_j+ f_{t,i}(\ms +\gamma_t- (X_{t-1})_i)\\
	&=\left((a_t)_{ii} -\dfrac{ f_{t,i}(\ms +\gamma_t - (X_{t-1})_i)}{(\ms +\gamma_t- (X_{t-1})_i)}\right)(X_{t-1}-\ms-\gamma_t)_i+ \sum_{j\neq i}^n (A_t)_{ij}(X_{t-1}-\ms)_j\\
	&=\left((a_t)_{ii} -f_{t,i}^\prime (\xi)\right)(X_{t-1}-\ms)_i+ \sum_{j\neq i}^n (A_t)_{ij}(X_{t-1}-\ms)_j  +\gamma_t f_{t,i}^\prime (\xi)\\
	&=\left((a_t)_{ii} -f_{t,i}^\prime (\xi)\right)(X_{t-1}-\ms)_i+ (1-(a_t)_{ii})(X_{t-1}-\ms)_j +\gamma_t f_{t,i}^\prime (\xi)\\
	&\leq\left(  |(a_t)_{ii}-f_{t,i}^\prime(\xi_i)|  |X_{t-1} -\ms|_i \right) +(1-(a_t)_{ii})|X_{t-1} -\ms|_j +|\gamma_t| f_{t,i}^\prime (\xi)\\
	&\leq\left(  |(a_t)_{ii}-f_{t,i}^\prime(\xi_i)|  + 1-(a_t)_{ii}\right) |X_{t-1} -\ms|_\infty+|\gamma_t| f_{t,i}^\prime (\xi) \\
	&\leq \sup_i \sup_\xi \left(  |(a_t)_{ii}-f_{t,i}^\prime(\xi_i)|  + 1-(a_t)_{ii}\right) |X_{t-1} -\ms|_\infty+C|\gamma_t|
\end{align*}
The rest of the proof follows as in the proof of Theorem~\ref{thm:blwn}, more precisely, following the same argument starting with \eqref{e:ll4}.  In all instances the convergence follows the same arguments as in the linear case.  

\item 
First observe that from \eqref{e:n10} we get 
\[
\E[|X_t|_\infty]\le \rho_{t}\E[|X_{t-1}|_\infty] +2\E[\gamma]. 
\]
From this, iterating and using \eqref{e:n4} as in the linear case we obtain that 
\[
\sup_{t\ge1}\E[|X_t|_{\infty}]=C<\infty.  
\]

To treat the case where $\gamma_{t}$ are all iid, we follow the same argument as the linear case.  Here we have to use in the first place the distance defined in \eqref{e:D} and the argument for the estimate of 
$D(X_t,X_{t-1})$ we need to take a for any coupling $\tilde{X}_{t-1}$ and $\tilde{X}_{t-2}$ the coupling $A_t\tilde{X}_{t-1}+f_t(\gamma-\tilde{X}_{t-1})$ and $A_{t-1}\tilde{X}_{t-2}+f_{t-1}(\gamma-\tilde{X}_{t-2})$.  Then, 
\[
\begin{split}
	D(X_t,X_{t-1})&\le \E[|A_t\tilde{X}_{t-1}+f_t(\gamma-\tilde{X}_{t-1})-A_{t-1}\tilde{X}_{t-2}+f_{t-1}(\gamma-\tilde{X}_{t-2})|_{\infty}] \\
	& \le \E[|A_t\tilde{X}_{t-1}+f_t(\gamma-\tilde{X}_{t-1})-(A_{t}\tilde{X}_{t-2}+f_{t}(\gamma-\tilde{X}_{t-2})|_{\infty}] \\
	& \quad +\E[|A_{t}\tilde{X}_{t-2}+f_{t}(\gamma-\tilde{X}_{t-2})-(A_{t-1}\tilde{X}_{t-2}+f_{t-1}(\gamma-\tilde{X}_{t-2}))|_\infty] \\
	& \le \rho_t\E[|\tilde{X}_{t-1}-\tilde{X}_{t-2}|_\infty +\left(|A_t-A_{t-1}|_\infty+\max_{i}\sup_{\xi\in\R}|f_{t,i}'(\xi)-f_{t-1,i}'(\xi)|\right)\E[|X_{t-2}|_\infty] \\
	&\le \rho_tD(X_{t-1},X_{t-2})+C\left(|A_t-A_{t-1}|_\infty+\max_{i}\sup_{\xi\in\R}|f_{t,i}'(\xi)-f_{t-1,i}'(\xi)|\right).
\end{split}
\]
From this we proceed exactly in the same way as in the proof of the linear case, more precisely, the same proof following \eqref{e:ll12} to show that $X_t$ is Cauchy in the metric $D$.

\end{enumerate}
\end{proof}

Notice that the last part of the result above does not involve $\bar{\sigma}$ because it is actually hidden in the sequence $\gamma$.  As opposed to the other two cases, the convergence is only in distribution and in principle that is implicitly defined, it is not a constant variable as in the previous cases.

\begin{remark}
	Matrix $A(t)$ and learning function $f_t$ are allowed to be time dependent or slowly varying. They could be random but in a controlled way. Were $A$ and $f$ to be fixed in time, the above result would still hold. So the constant case is a special case of what we have shown.
\end{remark}

Continuity of the learning function $f_t$ is essential. We give an example of a situation where it breaks down.

\begin{example}
	Consider the sign function $\sign(x)=1$ for $x>0$, $\sign(0)=0$ and $\sign(x)=-1$ for $x<-1$.  If the learning $f$ were the sign function, then the dynamics would be 	
	\[
	X_{t}=A_tX_{t-1}+ \mathcal{E}\sign(\ms - X_{t-1}).
	\] 
	Consensus in this circumstance would not be achieved.  One can plainly see this in the one dimensional case of $A_t=1$, $\sigma=1$, $Y_t=X_t-\sigma$, $Y_0=1$ and take $1/3<\mathcal{E}<1/2$.  With this setup we get 
	\[
	Y_1=1-\mathcal{E},Y_2=1-2\mathcal{E}, Y_3=1-3\mathcal{E}, Y_4=1-2\mathcal{E},Y_5=1-3\mathcal{E},\dots
	\]
	which shows that $Y_t$ becomes periodic, thus not convergent.  We can extend this behavior to more general situations of course, though this periodic pattern still follows.  
\end{example}

\section{Limit Theorems in Distribution for time invariant models} \label{s:CLT}
We study a CLT result in the case of constant $A$ and $\mathcal{E}$. The  analysis for  \autoref{thm:blwn} required condition \eqref{e:ll1}. However, the use of this condition in the dynamics or iterations means a CLT is not possible. To obtain a general CLT result, we have to change this condition. Moreover, we examine the dynamics in more general form. In the previous theorem, to ensure consensus, each agent had to interact and learn. While there may be periods of no learning $(\eps_t)_{ii}=0$ for some players in which case $\rho_{t} \geq 1$ for short bouts, eventually all agents have to learn and have positive self-belief $(a_t)_{ii}$.  To see this from a much higher perspective, we consider here the case of constant Average plus Learning matrices and study the limiting behavior of the beliefs in which the noise is added to the model.  Thus we have
\begin{equation}\label{e:clt0}
	X_{t+1}=AX_{t}+\mathcal{E}(\bar{\sigma}-X_t)+\gamma_t
\end{equation}
where the noise $\gamma$ is assumed iid.  In fact, we can continue with the model we studied above where the noise was inside the learning part, namely
\[
X_{t+1}=AX_{t}+\mathcal{E}(\bar{\sigma}-X_t+\gamma_t).
\]
However, within the assumption that $A$ and $\mathcal{E}$ are constant we can simply redefine $\tilde{\gamma}_t=\mathcal{E}\gamma_t$ and with this change the above equation becomes
\[
X_{t+1}=AX_{t}+\mathcal{E}(\bar{\sigma}-X_t)+\tilde{\gamma}_t
\]
which is essentially the model \eqref{e:clt0}. We take one more step and rewrite the equations \eqref{e:clt0} in the form
\begin{equation}\label{e:clt0a}
	X_{t+1}-\bar{\sigma}=(A-\mathcal{E})(X_{t}-\bar{\sigma})+\gamma_t.
\end{equation}


Assume that the matrix $A-\mathcal{E}$ has a standard Jordan form
\begin{equation}\label{e:Jordanform0}
  A-\mathcal{E}=P^{-1}JP
\end{equation}
where $J$ is the Jordan decomposition of $A-\mathcal{E}$ with the blocks $(J_k)_{k=1,\dots, l}$ on the diagonal and $J_k$ having dimension $m_k\times m_k$ and being defined by the eigenvalue $\lambda_k$.  Here we can take the complex Jordan decomposition or the real decomposition.  The computations are cleaner with the complex decomposition however the statements we are going to make are easily transferable to the real case as well.   

Now consider
\[
\begin{split}
	\alpha &=\max\{|\lambda_i|:i=1,2,\dots,l\} \text{ and } W=\{i\in \{1,2,\dots,l\}: |\lambda_i|=\alpha\}, \\ 
	m&=\max\{m_i:i \in W\} \text{ and } W_{max}=\{i\in W: m_i=m\}
\end{split}
\]
and set $\mathcal{W}=\cup_{i\in W}\{\sum_{j=1}^i m_j,\dots, (\sum_{j=1}^{j+1}m_{j})-1\}$ and similarly  
\[\mathcal{W}_{max}=\cup_{i\in W_{max}}\{\sum_{j=1}^i m_j,\dots, (\sum_{j=1}^{j+1}m_{j})-1\} \]
which represents the index set in $\{1,2,\dots,n\}$ corresponding to the Jordan blocks $J_i$ with $i$ in $W$ or $W_{max}$.  Denote by 

\begin{equation}
	B=P^{-1} J_W P \text{ and } B_{max}=P^{-1}J_{W_{max}} P
\end{equation} 
where  $J_{W}$ ($J_{W_{max}}$) is the block matrix where only the blocks with indices contained in $W$ (or $W_{max}$) appear, all the others having been replaced by 0.  

Furthermore, we also introduce the matrices 
\begin{equation}\label{e:clt12}
	D_W, D_{W_{max}} \text{ and } L_W, L_{W_{max}}
\end{equation}
as the diagonal of the matrices $J_W$, $J_{W_{max}}$ 
respectively as the off diagonal parts
of $J_W$, $J_{W_{max}}$.  

In addition to these we will consider the following matrices 
\begin{equation}
	Q_{W}=\frac{1}{\alpha}P^{-1} D_W P \text{ and } Q_{W}^{\#}=\alpha P^{-1} D_W^{\#} P 
\end{equation}
where the inverse $D_W^{\#}$ is defined as matrix with the inverses of the non-zero blocks.  In general $D_W^{\#}$ is not invertible and the non-zero elements on the diagonal of $\alpha D^{\#}$ are on the unit circle.  Similarly, we define 
\begin{equation}
	Q_{max}=\frac{1}{\alpha}P^{-1} D_{W_{max}} P \text{ and } Q_{max}^{\#}=\alpha P^{-1} D_{W_{max}}^{\#} P.
\end{equation}

As an example, assume that the Jordan block matrix $J$ is the following
\begin{equation}\label{e:JordanBlock1}
    J=\mat{\lambda_1 & 1 & 0&0&0  \\ 0 & \lambda_1 & 0 & 0&0 \\ 0&0&\lambda_2 &1 &0 \\ 0&0&0&\lambda_2 &0   \\ 0&0&0&0&\lambda_3 & }
\end{equation}
where $\lambda_1,\lambda_2,\lambda_3$ are complex numbers. Clearly there are three Jordan blocks. 

In this case, if $|\lambda_1|>\max\{|\lambda_2|,|\lambda_3|\}$, then $\alpha=|\lambda_1|$ and the quantities described above become $W=\{1\}, m=2, W_{max}=\{1\}$ and 
\[
J_W=J_{W_{max}}=\mat{\lambda_1 & 1 & 0&0&0 \\ 
0 & \lambda_1 & 0 & 0&0 \\ 
0&0&0 &0 &0\\ 
0&0&0 &0 &0\\ 
0&0&0 &0 &0
},
L_W=L_{W_{max}}=\mat{0 &1 & 0&0&0  \\ 
0 & 0 & 0 & 0&0 \\ 
0&0&0 &0 &0\\ 
0&0&0 &0 &0\\ 
0&0&0 &0 &0
}
\]
with 
\[
D_W=D_{W_{max}}=\mat{\lambda_1 &0 & 0&0&0 \\ 0 & \lambda_1 & 0 & 0&0 \\ 
0&0&0 &0 &0\\ 
0&0&0 &0 &0\\ 
0&0&0 &0 &0
}, \alpha D_W^{\#}=\alpha D_{W_{max}}^{\#}=\mat{\alpha/\lambda_1 &0 & 0&0&0 \\
0 & \alpha/\lambda_1 & 0 & 0&0 \\ 
0&0&0 &0 &0\\ 
0&0&0 &0 &0\\ 
0&0&0 &0 &0
}.
\]

In the case we have $|\lambda_1|=|\lambda_2|>|\lambda_3|$, then we obtain 
$W=\{1,2\}, m=2, W_{max}=\{1,2\}$ and \[
J_W=J_{W_{max}}=\mat{\lambda_1 & 1 & 0&0&0 \\ 
0 & \lambda_1 & 0 & 0&0 \\ 
0&0&\lambda_2 &1 &0\\ 
0&0&0 &\lambda_2 &0\\ 
0&0&0 &0 &0
},
L_W=L_{W_{max}}=\mat{0 &1 & 0&0&0  \\
0 & 0 & 0 & 0&0 \\ 
0&0&0 &1 &0\\ 
0&0&0 &0 &0\\ 
0&0&0 &0 &0
}
\]
with 
\[
D_W=D_{W_{max}}=\mat{\lambda_1 &0 & 0&0&0 \\ 0 & \lambda_1 & 0 & 0&0 \\ 
0&0&\lambda_2 &0 &0\\ 
0&0&0 &\lambda_2 &0\\ 
0&0&0 &0 &0
}, \alpha D_W^{\#}=\alpha D_{W_{max}}^{\#}=\mat{\alpha/\lambda_1 &0 & 0&0&0 \\ 0 & \alpha/\lambda_1 & 0 & 0&0 \\ 
0&0&\alpha/\lambda_2 &0 &0\\ 
0&0&0 &\alpha/\lambda_2 &0\\ 
0&0&0 &0 &0
}.
\]

In the case we have $|\lambda_3|>\max\{|\lambda_1|,|\lambda_2|\}$, then we obtain 
$W=\{3\}, m=1, W_{max}=\{3\}$ and \[
J_W=J_{W_{max}}=\mat{0&0 & 0&0&0 \\ 
0 & 0 & 0 & 0&0 \\ 
0&0&0 &0 &0\\ 
0&0&0 &0 &0\\ 
0&0&0 &0 &\lambda_3
},
L_W=L_{W_{max}}=\mat{0 &0 & 0&0&0  \\
0 & 0 & 0 & 0&0 \\ 
0&0&0 &0 &0\\ 
0&0&0 &0 &0\\ 
0&0&0 &0 &0
}
\]
with 
\[
D_W=D_{W_{max}}=\mat{
0&0&0 &0 &0 \\ 
0&0&0 &0 &0 \\ 
0&0&0 &0 &0\\ 
0&0&0 &0 &0\\ 
0&0&0 &0 &1
}, \alpha D_W^{\#}=\alpha D_{W_{max}}^{\#}=\mat{0 &0 & 0&0&0 \\ 0 & 0 & 0 & 0&0 \\ 
0&0&0 &0 &0\\ 
0&0&0 &0 &0\\ 
0&0&0 &0 &\alpha/\lambda_3
}.
\]

Now we can state the main result of this section. In \textit{Case I}, when $\alpha<1$, we have convergence to a distribution. There is no scaling. In the original model of learning with non-varying averaging matrix and learning rates,  in the case, $\mathcal{E}>0$, the iterations have a contraction and thus this result is similar in spirit to the case of the time varying case with persistent noise and the fourth item of \autoref{thm:blwn}.  However we provide here an explicit description of the limiting distribution. 
In the theorem below we consider the other two cases of $\alpha>1$ and $\alpha=1$. It should be pointed out that \textit{Case III} covers a pure DeGroot model with no learning: $\mathcal{E}=0$. For \textit{Case II}, the interpretation could be that the players are putting too much weight on their feedback or learning ability or the $L^1$ norm of $A-\mathcal{E}$ is greater than 1. Thus, \autoref{thm:blwn} and \autoref{t:clt} complement each other by addressing different features of Averaging plus Learning. 


\begin{theorem}\label{t:clt}
Assume that $X_t$ satisfies
\[
X_{t+1}=(A-\mathcal{E})X_{t}+\gamma_t, \text{ for }t\ge0,
\]
where $(\gamma_t)_{t\ge0}$ is an iid sequence of random variables.   
Using the decomposition \eqref{e:Jordanform0}, we denote 
\[
\alpha =\max\{|\lambda_i|:i=1,2,\dots,l\}
\] where $l$ is the number of Jordan blocks and $\lambda_i$ is the eigenvalue of each Jordan block. 
Then,
\begin{enumerate}
	\item[Case I.] If $\alpha<1$ and the noise $\gamma_t$ is in $L^1$,  then
	\begin{equation}\label{e:clt1}
		X_t \Longrightarrow  \sum_{s\ge0}(A-\mathcal{E})^s\gamma_s.
	\end{equation}
	\item[Case II.] If $\alpha>1$ and the noise $\gamma_t$ is in $L^1$, then we can write 
	\begin{equation}\label{e:clt2}
		\frac{X_t}{t^{m-1}\alpha^t}= \frac{L_{max}^{m-1}Q_{max}^{t-m+1}}{\alpha^{m-1}(m-1)!}\left(X_0+\sum_{s\ge 0}\frac{(Q_{max}^{\#})^s}{\alpha^s}\gamma_{s}\right) + R_t
	\end{equation}
	where $R_t$ converges to $0$ in $L^1$.
	\item[Case III.] If $\alpha=1$, and the noise $\gamma_t$ is $L^2$, iid with mean $\mu$ and covariance matrix $\Gamma$, then
	\begin{equation}\label{e:clt3}
		\frac{X_t-\E[X_t]}{t^{m-1/2}} \Longrightarrow N(0,C)
	\end{equation}
	where the convergence is in weak/distribution sense and the covariance matrix $C$ is given by 
	\[
	C=\frac{1}{(2m-1)(m-1)!^2} K^{m-1}\Gamma (K^{m-1})^T \text{ with }K=P^{-1}L_{W_{max}}P.
	\]
	with the important convention that for $m=1$, $K^0=Id$ for any matrix $K$ (including the zero matrix).  
\end{enumerate}
\end{theorem}
\begin{proof} Before we start the proof we point out that the key to the analysis here is the Jordan decomposition.  We will use here the convention that the Jordan blocks are real valued, however we use the complex version for the sake of the exposition.  The real case can be worked out in a similar fashion with a little bit more care of the algebra.  For a clarification, we point out that the real Jordan decomposition can be realized from the complex decomposition.  
	
	Using the decomposition $A-\mathcal{E}=P^{-1}JP$ and denoting $\tilde{X}_t=PX_t$, then we get 
	\[
	\tilde{X}_{t+1}=J \tilde{X}_{t}+P\gamma_t.
	\]
	Because the matrix $J$ is a block diagonal matrix, we can reduce the analysis to each block. The general result then follows by transferring the results to $X_t=P^{-1}\tilde{X}_t$.    
	
	We will treat each case separately for each Jordan block.  In this case we fix an index $k$ and write $J=J_k$ as  
	\[
	J_k=\met{\lambda &1  &0 & \hdotsfor{3} \\ 0 & \lambda & 1 & 0 & \hdotsfor{2}\\ \hdotsfor{6}\\ \hdotsfor{6} \\ 0&\hdotsfor{2} & 0 &\lambda & 1 \\ 0&\hdotsfor{3} & 0 & \lambda}=\lambda Id+ L
	\]
	where $\lambda=\lambda_k$ and 
	\[
	L=\met{0 &1  &0 & \hdotsfor{3} \\ 0 & 0 & 1 & 0 & \hdotsfor{2}\\ \hdotsfor{6}\\ \hdotsfor{6} \\ 0&\hdotsfor{2} & 0 & 0 & 1 \\ 0&\hdotsfor{3} & 0 & 0}. 
	\]
	To simply the exposition here we assume that $m=m_k$ is the dimension of the Jordan block.    
	
	Notice the important fact that $L^m=0$. The key property which follows from this is that 
	\begin{equation}\label{e:clt4}
		J^i=\sum_{j=0}^{i\wedge (m-1)}\lambda^{i-j}{i \choose j}L^j,
	\end{equation}
	with the convention that $a\wedge b=\min\{a,b \}$.

	Now we consider the first case of our result, namely $\alpha<1$.  In this case, we denote by $Y_t$  the coordinates of $\tilde{X}_t$ corresponding to the Jordan block $J$ and let also $\delta_t$ the same corresponding part of $P\gamma_t$.  Then, an easy algebraic calculation gives (cf. \eqref{e:clt4}) that
	\[
	\begin{split}
		Y_{t}&=J^tY_0+\delta_{t-1}+J\delta_{t-2}+\dots+J^{t-1}\delta_0   \\
		& = J^{t}Y_0 +\sum_{s=0}^{t-1}J^{s}\delta_{t-1-s} .
	\end{split}
	\]
	Because $|\lambda|<1$, it is not difficult to observe (essentially from \eqref{e:clt4}) that $J^t$ converges to $0$ as $t$ converges to $\infty$.  In particular what this means is that $J^{t}Y_0$ converges to $0$ as $t$ converges to infinity.  On the other hand, the second term, namely $\sum_{s=0}^tJ^{s}\delta_{t-s} $ has the same distribution as $\sum_{s=0}^t J^{s}\delta_{s}$.  This sum is convergent in $L^1$.  Indeed if $\lambda=0$, then $J^t=0$ for $t\ge m$.  On the other case if $0<|\lambda|<1$ then, for example using \eqref{e:clt4}, we can assure that $\alpha^{-t}J^s$ is a bounded matrix, thus in particular we obtain that $|J^s|\le C\alpha^s$ for all $s\ge 1$, consequently it is now an elementary task to show that the series $\sum_{s=0}^\infty J^{s}\delta_{s}$ is convergent in $L^1$. Putting all the pieces together, we can easily see the conclusion. 
	
	For the second case, $\alpha>1$, we use again a reduction to blocks analysis.  For a given block, we write now as above
	\[
	\begin{split}
		Y_{t}&=J^tY_0+\delta_{t-1}+J\delta_{t-2}+\dots+J^{t-1}\delta_0   \\
		& = J^{t}Y_0 +\sum_{s=0}^{t-1}J^{t-1-s}\delta_{s}. 
	\end{split}
	\]
	As opposed to the previous case when $|\lambda|<1$ we now look at the 
	\begin{equation}\label{e:clt6}
		\frac{Y_{t}}{t^{m-1}\lambda^t} = \frac{J^{t}Y_0}{t^{m-1}\lambda^t} +\sum_{s=0}^{t-1}\frac{J^{t-1-s}\delta_{s}}{t^{m-1}\lambda^t}.
	\end{equation}
	From the above expression, there are two terms we need to take into account.  Now, using \eqref{e:clt4} for a given $s=0,1,2,\dots$ we analyze the asymptotic of  
	\[
	\frac{J^{t-s}}{t^{m-1}\lambda^t}=\sum_{j=0}^{(t-s)\wedge (m-1)}\frac{\lambda^{t-s-j}{t-s \choose j}}{t^{m-1}\lambda^t}L^j.
	\]
	The point is that we have a finite number of terms in the above sum and we can take the limit as $t\to\infty$ for each individual term.  For instance we have 
	\[
	\frac{\lambda^{t-s-j}{t-s \choose j}}{t^{m-1}\lambda^t}\xrightarrow[t\to\infty]{}
	\begin{cases}
		0 & \text{ if } j<m-1 \\ 
		\frac{1}{\lambda^{s+m-1}(m-1)!} & \text{ for } j=m-1.  
	\end{cases}
	\]
	From this we derive that for a fixed $0\le s \le t-1$, 
	\[
	\frac{J^{t-1-s}}{t^{m-1}\lambda^t}\xrightarrow[t\to\infty]{}
	\frac{1}{\lambda^{s+m}(m-1)!}L^{m-1}. 
	\]
	Now, we go back to \eqref{e:clt6} and notice that 
	\begin{equation}\label{e:clt7}
		\frac{J^{t}Y_0}{t^{m-1}\lambda^t}\xrightarrow[t\to\infty]{} \frac{L^{m-1}Y_0}{\lambda^{m-1}(m-1)!}
	\end{equation}
	Now we are going to split the series from \eqref{e:clt6} as 
	\begin{equation}\label{e:clt8}
		\sum_{s=0}^{t-1}\frac{J^{t-1-s}\delta_{s}}{t^{m-1}\lambda^t}=\sum_{s=0}^{u}\frac{J^{t-1-s}\delta_{s}}{t^{m-1}\lambda^t}+\sum_{s=u+1}^{t-1}\frac{J^{t-1-s}\delta_{s}}{t^{m-1}\lambda^t}
	\end{equation}
	where $u$ is a fixed but large number such that $0<u<t-1$.  For the first part of the series we have that for a fixed $u$,
	\begin{equation}\label{e:clt10}
		\lim_{t\to\infty}\sum_{s=0}^{u}\frac{J^{t-1-s}\delta_{s}}{t^{m-1}\lambda^t}= \frac{L^{m-1}}{\lambda^{m}(m-1)!}\sum_{s=0}^{u}\frac{\delta_s}{\lambda^s}. 
	\end{equation}
	The second term can be controlled as follows. Take any matrix norm and estimate 
	\[
	\E[\|\sum_{s=u+1}^{t-1}\frac{J^{t-1-s}\delta_{s}}{t^{m-1}\lambda^t} \|]\le \sum_{s=u+1}^{t-1}\frac{[\|J^{t-1-s}\| \, \E|\delta_{s}|]}{t^{m-1}\lambda^t} \le c\E[|\delta_0|]\sum_{s=u+1}^{t-1}\frac{1}{|\lambda|^s}\le \frac{C}{|\lambda|^{u+1}}
	\]
	where $c,C>0$ are some constants independent of $u$ and $t$.  Using now \eqref{e:clt8}, \eqref{e:clt10} and \eqref{e:clt10} we can conclude that 
	\[
	\limsup_{t\to\infty}E[|\sum_{s=0}^{t-1}\frac{J^{t-1-s}\delta_{s}}{t^{m-1}\lambda^t}-\frac{L^{m-1}}{\lambda^{m}(m-1)!}\sum_{s=0}^{u}\frac{\delta_s}{\lambda^s}|]\le \frac{C}{|\lambda|^{u+1}}.
	\]
	Now as the series $\sum_{s=0}^{\infty}| \frac{\delta_s}{\lambda^s}|$ is convergent in $L^1$ for $|\lambda|>1$ which means that we can let $u$ tend to infinity and get the conclusion that 
	\[
	\frac{Y_{t}}{t^{m-1}\lambda^t} = \frac{L^{m-1}}{\lambda^{m-1}(m-1)!}
	\left( Y_0+\sum_{s=0}^{\infty}\frac{\delta_{s}}{\lambda^s}\right)+R_t.
	\]
	where the remainder $R_t$ is a random variable such that $E[|R_t|]\xrightarrow[t\to\infty]{}0$. 
	
	For each eigenvalue $\lambda$ with $|\lambda|=\alpha$, we can write it as $\lambda=\alpha e^{i\theta}$ for some $\theta\in[0,2\pi)$ and with this representation we now have 
	\[
	\frac{Y_{t}}{t^{m-1}\alpha^t} = \frac{e^{i(t-m+1)\theta}L^{m-1}}{\alpha^{m-1}(m-1)!}
	\left( Y_0+\sum_{s=0}^{\infty}\frac{e^{-is\theta}\delta_{s}}{\alpha^s}\right)+R_t.
	\]
	Putting all the contributing blocks together, we get second part of the Theorem.

	For the last part, namely Case III of the Theorem, for $\alpha=1$ we can simply use a multidimensional version of the CLT.  For a single Jordan block we have 
	\[
	Y_t-\E[Y_t] = \delta_{t-1}-E[\delta_{t-1}]+J(\delta_{t-2}-\E[\delta_{t-2}])+\dots+J^{t-1}(\delta_0 -\E[\delta_0]).
	\]
	which in distribution is the same as 
	\[
	Y_t-\E[Y_t]\sim \delta_{0}-E[\delta_{0}]+J(\delta_{1}-\E[\delta_{1}])+\dots+J^{t-1}(\delta_{t-1} -\E[\delta_{t-1}])
	\]
	Using \cite[Theorem 2.3.8]{Stroock} we need first to compute the covariance matrix 
	\[
	\Lambda_t=\sum_{s=0}^{t-1} Cov(J^{t-1-s}\delta_s)=\sum_{s=0}^{t-1}J^{t-1-s}\Gamma (\overline{J^{t-1-s}})^T,
	\]
	where $\Gamma$ is the covariance matrix of $\delta_s$ and we use the bar here to denote the complex conjugate.  Next, we write 
	\[
	J^{k}=\sum_{j=0}^{k\wedge (m-1)}\lambda^{k-j}{k \choose j}L^j=\lambda^k\sum_{j=0}^{m-1} P_j(k)
	\]
	where $P_j$ is a matrix valued polynomial of degree $j$.  The coefficient of $P_{m-1}$ is $\frac{1}{(m-1)!}L^{m-1}$ and in general we have 
	\[
	|J_t|\le C t^{m-1} 
	\]
	for some constant $C>0$ which does not depend on $t$.  
	In particular we have that 
	\[
	\Lambda_t=\sum_{j,l=0}^{m-1}\sum_{s=0}^{t-1}P_j(s)\Gamma(\overline{P_l(s)})^T.
	\]
	The leading term in $t$ of the above expression is given by the polynomials of the largest degrees, thus 
	\[
	\Lambda_t=\sum_{s=0}^{t-1}\frac{s^{2m-2}}{(m-1)!^2}L^{m-1}\Gamma (L^{m-1})^T +O(\sum_{s=0}^{t-1}s^{2m-4}) = \frac{t^{2m-1}}{(2m-1)(m-1)!^2} L^{m-1}\Gamma (L^{m-1})^T +O(t^{2m-3}).
	\]
	To use the CLT we need to check the Lindeberg condition, namely \cite[Equation (2.3.10)]{Stroock}.  This is now easily done by observing that 
	\[
	\E[|J^{s}\delta_s|^2,|J^{s}\delta_s|>\epsilon \Lambda_t]\le  Cs^{m-1}\E[|\delta_s|^2,|\delta_s|>\epsilon \frac{\Lambda_t}{Cs^{m-1}}] \le C\frac{s^{2m-2}}{\epsilon^2\Lambda_t^2}
	\]
	for a constant $C>0$ independent of $s$.  Summing this over $s$, we see it is easy to verify now the Lindeberg condition. From this, we get the conclusion by putting together all the blocks and noticing that the eigenvalues $\lambda$ with $|\lambda|<1$ do not contribute anything if we scale $X_t$ by $t^{m-1/2}$.   
	
\end{proof}
Before we jump into the proof (in the Appendix), let us make some comments on the significance of this Theorem.  

\subsection{Significance and interpretation of Theorem~\ref{t:clt}} \label{thm:importance}

We highlight some important points of the result above which might shed some light on its relevance and importance: 

\begin{enumerate}
\item The first part of the Theorem is pretty straightforward and it should not come as a surprise given that we treated something like this in Theorem~\ref{thm:blwn}.  However this is slightly stronger than the previous result because we have a complete description of the limiting distribution, not mere existence.   
\item  The second item is very interesting.  It shows that for the case of large eigenvalues, the leading order is $t^{m-1}\alpha^t$.  On the other hand this is not convergent if we have complex eigenvalues, because the term $Q_{max}^{t-m+1}$ oscillates.  For example, in the case that all the eigenvalues are simple, this term is of the type $a\cos(t\theta)+b\sin(t\theta)$.  It is also interesting to point out that for each such complex eigenvalue we obtain an oscillatory term.  
\item The last item shows that we get a CLT, however this is very sensitive to the change of the matrix $A$. For the CLT, the contribution of the noise comes only through generic properties of $\gamma$, like the mean and covariance, however, in the first two items the whole noise is present in the asymptotic behavior. 

\item  If we keep in mind that in fact $X_t$ in this Theorem should be thought as $X_t-\bar{\sigma}$, then it becomes obvious that in general, knowing the matrix $A-\mathcal{E}$, we should be able to get statistical estimates for $\sigma$. However, as the first two items show, the noise is a contributing part of the asymptotic behavior. 
\item In some cases, it might happen that the asymptotic limits in \eqref{e:clt1} \eqref{e:clt2} and \eqref{e:clt3} might be 0.  In this case, what we can do is go back to the Jordan block and refine the estimates.    
\item We can see the fragility of such types of results. A slight change in the matrix $A$ could lead to radically different behaviors for the dynamical system.  The important lesson is that for understanding the limiting behavior, the distribution of the noise is an integrated part, except the rather rare case when we can see a CLT for which only the covariance matrix of the noise contributes to the limiting distribution.  
\item The CLT incorporates a richer structure than possible in just standard DeGroot learning. Because of the feedback term, Case III encapsulates a basic DeGroot model. With  $\alpha=1$, it is possible all the learning rates are zero, $\mathcal{E}=0$ and $A-\mathcal{E}$ in \ref{e:clt0} becomes just $A$, then we have a pure \textbf{noisy} DeGroot model with no learning:
$
X_{t+1}=AX_{t}+\gamma_t, \text{ for }t\ge0.
$
Alternatively, maybe some agents are \textit{not} learning but interacting only. In that case, $\mathcal{E}$ is of a lower rank and $\mathcal{E}\neq 0$. In both situations we allow for negative weights in $A$ and in $\mathcal{E}$, as long as Case III applies and the original weights matrix $A-\mathcal{E}$ was stochastic. 
\end{enumerate}

\subsection{Simulation to Convergence and Numerical Discussion}\label{s:cltsim}

Figures \ref{f:alpha1andbigCLT1} examines asymptotic case for five agents with binomial $\pm 1$ noise or uniform noise on $[-1,1]$ and the dynamical matrix given by \ref{e:JordanBlock1}.  We consider the case of noise with independent components on the left hand column of Figure~\ref{f:alpha1andbigCLT1} and the case the noise is rotated by an orthogonal matrix in the right hand column. The simulations contain several examples, depending on the values of  $\alpha$ and we plotted the pairwise distributions of the agents.  

If $\alpha<1$ with binomial noise we get some Cantor like distributions.  For other types of noise, this is not the case of course.  Notice that the Cantor type structure appears also in the case of noise with correlated components of the noise.  It is also interesting to point out that the limiting distribution heavily depends on the structure of the noise.  

In the case $\alpha>1$, we have collapses of the limiting distribution.  The dimension of the collapse is defined by the dimensions whose corresponding eigenvalues are $\alpha$.  In figure (c) and (d), we started with a random initial value of $X_0$ with components uniform and independent in $[0,10]$.  Notice that the limiting distribution is in fact depending not only on the noise but also on the in the initial condition, thus making this limiting case more intricate.  On the other hand, the dimension of the collapse (as it is described in \eqref{e:clt2}) is dictated by the subspace generated by the components with eigenvalues corresponding to $\alpha$.  Namely in Figure (c), the collapse is two dimensional while in Figure (d) is one dimension.  This is reinforced by the scale at which the other cases live on.  

The last case, which is the CLT, is $\alpha=1$ and it shows the limiting normal also lives on a lower dimensional space which is again defined by the eigenvalues of absolute values $1$.  This is clarified also by the fact that the covariance matrix is defined in this fashion.  For example, in the case of Figure (e) and (f), the covariance matrices of the limiting distribution are given 
\begin{equation}
C_{(e)}=\met{2/9&0&0&0&0 \\ 0&0&0&0&0 \\ 0&0&0&0&0 \\0&0&0&0&0 \\0&0&0&0&0} \text{ and }C_{(f)}=\met{1&0&0&0&0 \\ 0&0&0&0&0 \\ 0&0&1&0&0 \\0&0&0&0&0 \\0&0&0&0&0}. 
\end{equation}

Notice also here that in figure (e) $\lambda_1=1$ and $\lambda_3=1$ but the last component also disappears which at first site looks strange. This is in fact due to  the scaling in the CLT which is much bigger than the standard CLT scaling. Ultimately, this is the reflection of the fact that the first Jordan block is bigger than the last Jordan block.  

Thus, what these indicate is that the second component also vanishes in the limit, a fact which is somehow surprising but supported by the simulations as well.  One last thing is that the rotation matrix applied to the noise does not have any effect on the limiting distribution for the CLT, as opposed to the other two cases because in the end, the contribution of the noise is via the covariance matrix of the noise not through the whole distribution of the noise.

When $\alpha >1$, as in some agents are positively correlated, whereas others are anti-correlated. In this case, there is no dimensional collapse. However, when we have a pure DeGroot model, \autoref{f:alpha1andbigCLT1} (b), all five agents synchronize and reach stochastic consensus. In both of the \autoref{f:alpha1andbigCLT1} (a) and (b), the noise is not rotated. Furthermore, in $A$ the size of the largest Jordan block is 1. 

\begin{figure}[ht!]%
\centering
\subfloat[\centering non-CLT: $\alpha <1$]{{\includegraphics[width=6.15cm]{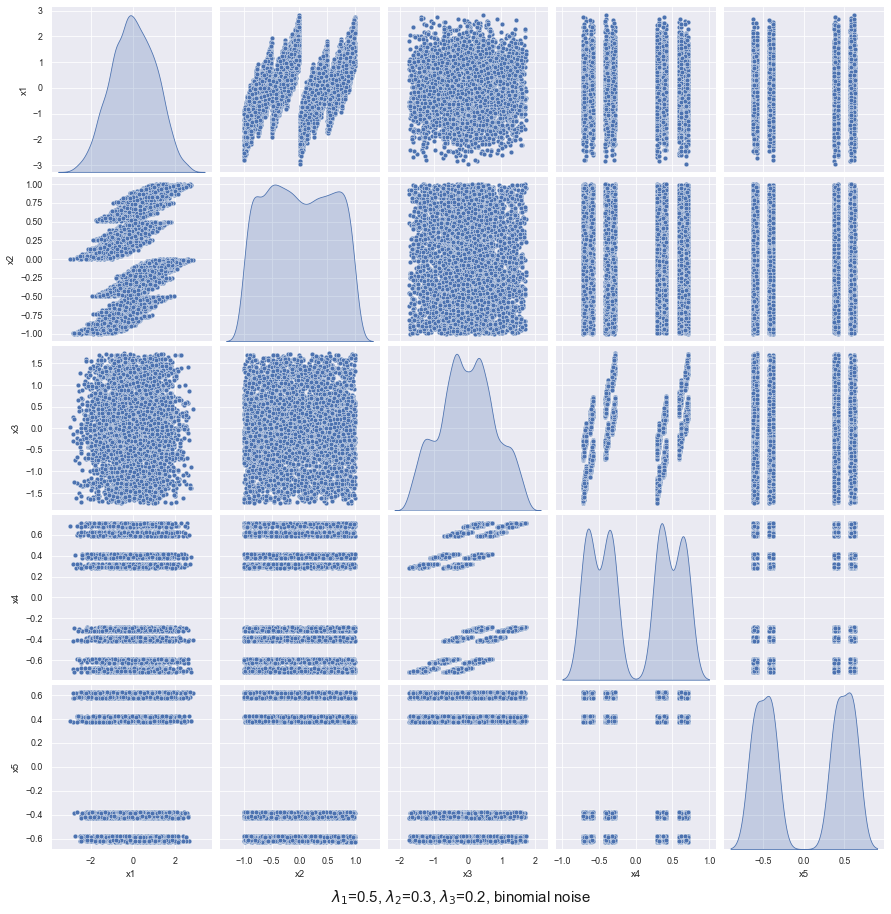} }}%
\qquad
\subfloat[\centering non-CLT: $\alpha <1$]{{\includegraphics[width=6.15cm]{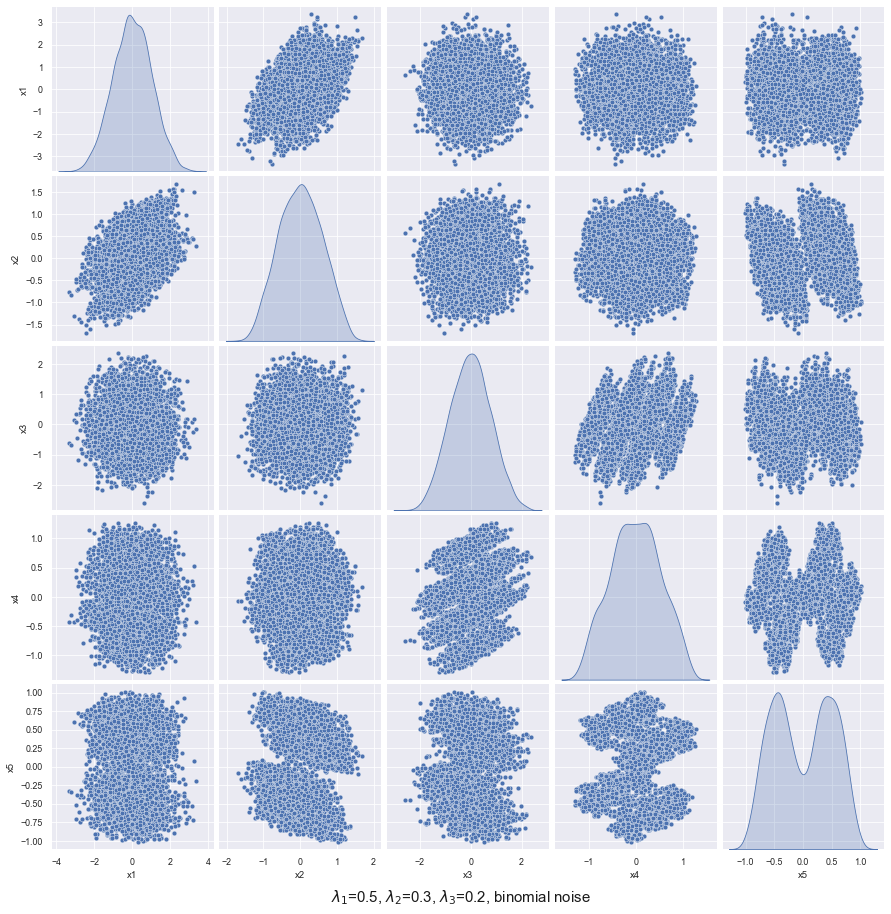} }}%
\qquad

\subfloat[\centering non-CLT: $\alpha >1$]{{\includegraphics[width=6.15cm]{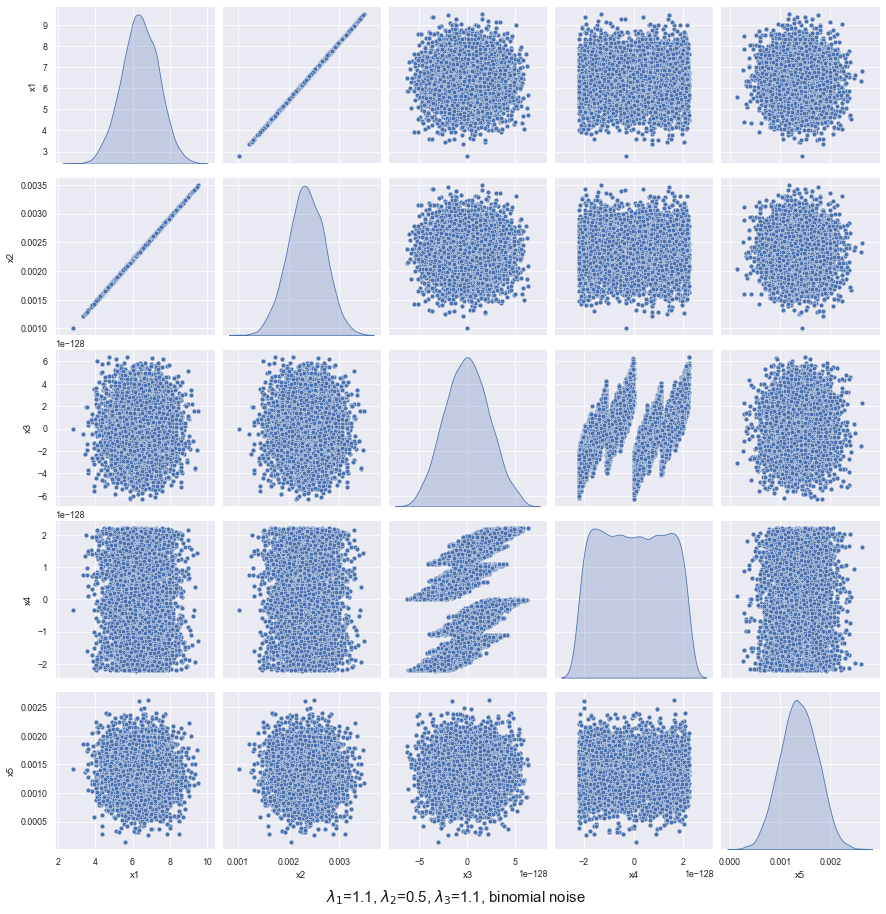} }}%
\qquad
\subfloat[\centering non-CLT: $\alpha > 1$]{{\includegraphics[width=6.15cm]{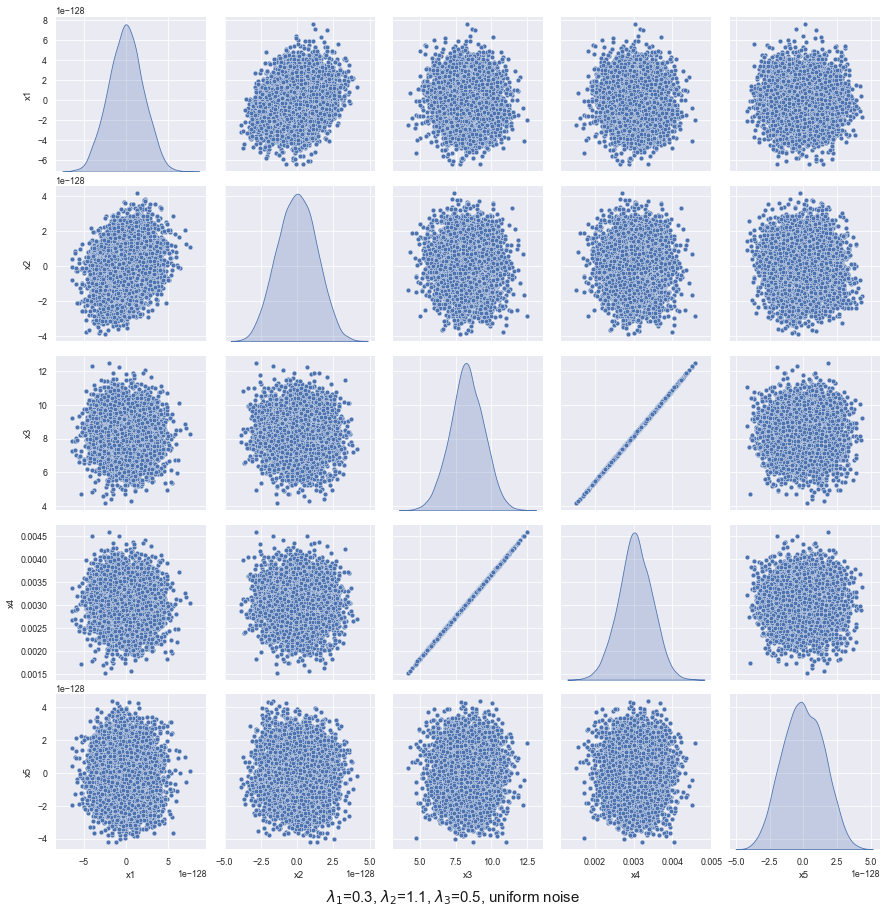} }}%
\qquad
\subfloat[\centering CLT: $\alpha =1$]{{\includegraphics[width=6.15cm]{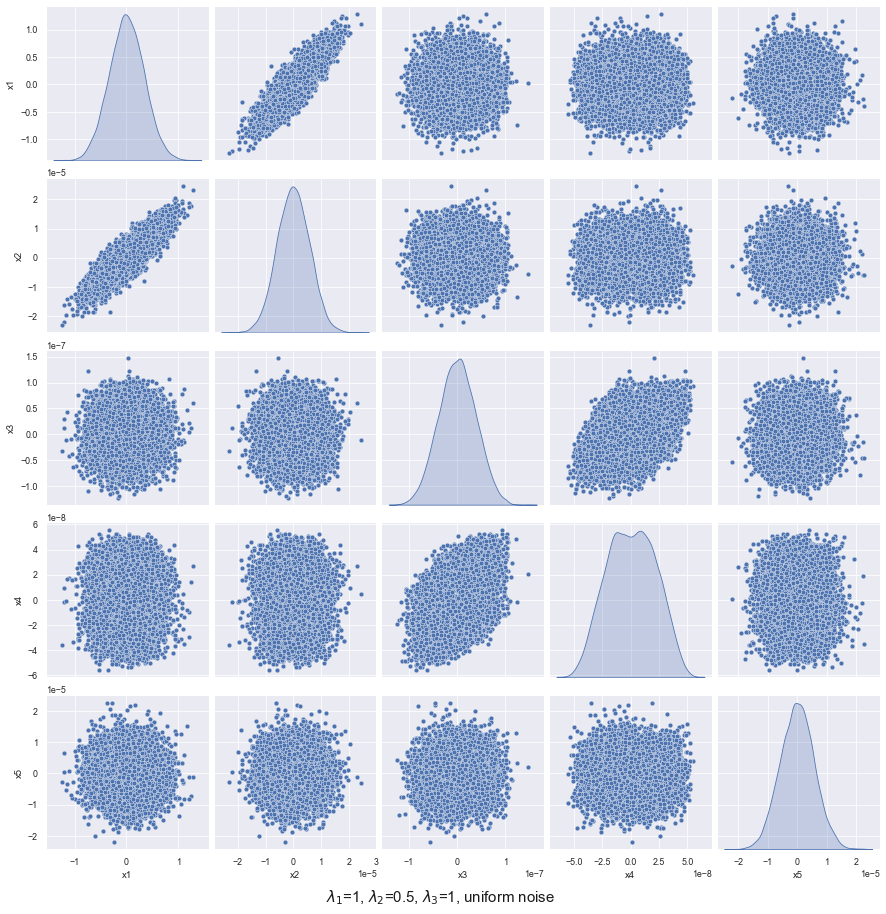} }}%
\qquad
\subfloat[\centering CLT: $\alpha =1$]{{\includegraphics[width=6.15cm]{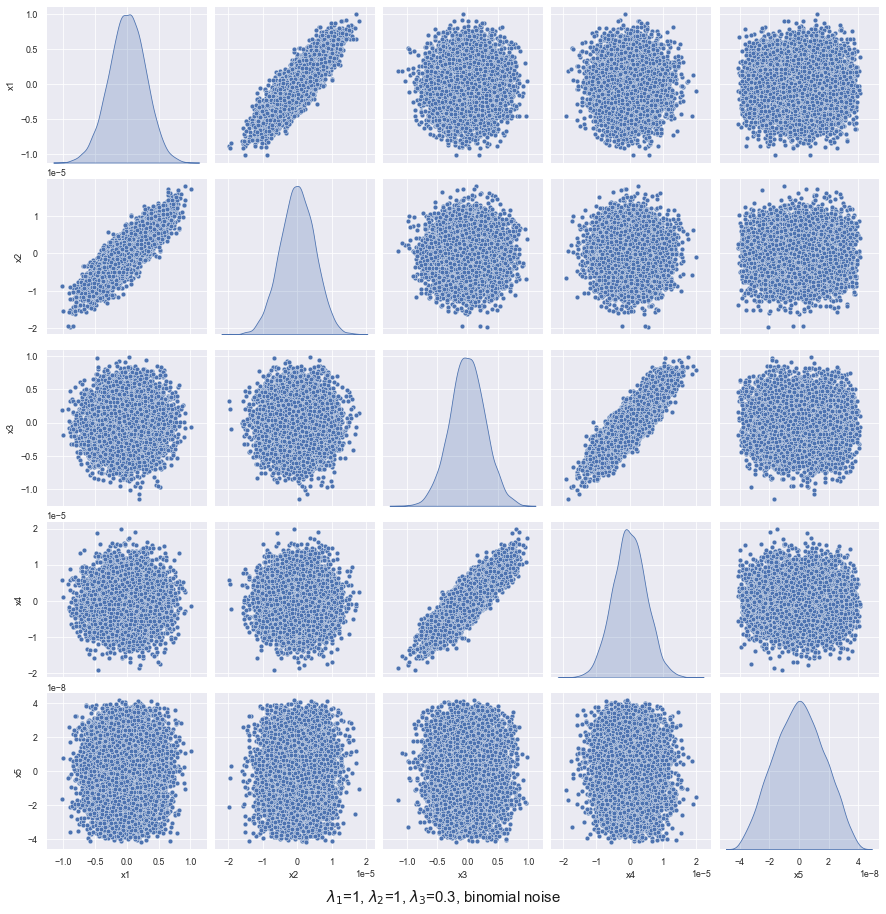} }}%
\caption{Figures (a) and (b) represent models with $\alpha<1$ which is not a CLT.   The second row, figures (c) and (d) represent the case of $\alpha>1$.  Notice the dimensional collapse in the components with eigenvalues $<1$.  In picture (c) the support of the distribution is $2$-dimensional while in figure (d) the support is $1$ dimensional.  In figures (e) and (f) we see a CLT limiting distribution.  Notice also the limiting distribution in (d) lives on a line while in (f) is two dimensional, the surviving components being the first and the third.   
}\label{f:alpha1andbigCLT1}
\end{figure}
To examine a more refined Jordan block structure, we directly use \ref{e:JordanBlock1} as our $A$ and vary $\lambda_1,\lambda_2$ and $\lambda_3$. The noise though iid, can be separated in two regimes: orthogonal or identity rotations. 
If the noise is rotated, the dynamics changes from adding just a pure $\gamma_t$ term to adding term $U\, \gamma_t$, where $U$ is an orthogonal matrix. When $U=I$, there is no rotation as in \autoref{f:alpha1andbigCLT1}. The rotational cases worth examining are the fractal or cantor-like distributions that arise in the $\alpha < 1$ case.

\section{Conclusion}


In this paper we introduce a model of interaction and learning which resembles to some extent the DeGroot model, though it is rather different in scope.  The plain vanilla model is deterministic but we introduce a model with randomness and analyze convergence when the noise decays to zero in large time.   

When the noise is not decaying, condition \eqref{e:ll1b} is crucial to ensure convergence in distribution. This condition can be thought of as a stabilization feature of learning. Individuals learn with varying $A_t$ and $\mathcal{E}_t$ but these cannot change too drastically. Eventually, all agents settle down.  
Temporary bubbles are also possible, where agents don't learn $\rho_t=0$ for shout bouts or there are periods of insanity $\rho_t >1$: this is an added feature of condition \eqref{e:ll1}. 

The limit theorem developed in Theorem \ref{t:clt} shows an intriguing phenomenon for the case of time independent matrix. Essentially, if we want to see the more refined structure of the $X_t$ \textemdash the opinions of the agents at time $t$\textemdash then the point is that the asymptotic behavior depends on the Jordan decomposition.  In some instances, we can get a CLT. However, thinking in terms of the matrix $A$ of the dynamics, this is rather unlikely.  On the other hand, if the eigenvalues stay inside the unit disk or are outside the unit disk, the main asymptotic limit depends, in fact, on the whole distribution of the noise. The more refined version of the results in Theorem \ref{t:clt} for the case of time varying dynamics matrix $A_t$ is desired, but given the fragility of the time independent case, a unitary approach seems more intricate.  Of value would be a treatment in which the matrix $A_t$ is picked out at random with some distribution.  Some of these topics appear in vastly different frameworks to ours \cite{diaconis1999iterated} and \cite{bhattacharya2003random}.   

Thus far, agents' rules are mechanical. Future work should address the issue of rationality or an infinite number of agents. In Averaging plus Learning, individuals are boundedly rational. They use the same rule. What if the agents are strategic? In the presence of noise or disturbance, manipulation of opinion dynamics by forceful agents like an AI bot (\cite{acemoglu2010spread}) becomes an interesting but difficult question.  Random dynamical systems were reviewed by \cite{bhattacharya2003random,stachurski2009economic}. Our results use different techniques to study social learning. Though it must be acknowledged that recursive random dynamical systems are not new in economics, physics and computer science, their probabilistic analysis poses several challenges to researchers.


\bibliographystyle{alpha}
\bibliography{ReferencesAPL17}

\end{document}